\documentclass[11pt]{article}
\usepackage{chicago}
\usepackage{times}

\usepackage{amssymb}

\usepackage{amsmath}
\usepackage{amsfonts}
\usepackage{amssymb}
\usepackage{amsthm}
\usepackage{graphicx}
\newcommand{\strat}{\mathbf{s}}
\newcommand{\SSigma}{S}
\newcommand{\ssigma}{s}
\newcommand{\G}{\mathcal{G}}
\newcommand{\NSD}{\mathit{NSD}}
\newcommand{\EU}{\mathit{EU}}
\newcommand{\supp}{\text{supp}}

\begin{document}


\newtheorem{theorem}{Theorem}[section]
\newtheorem{definition}[theorem]{Definition}
\newtheorem{fact}[theorem]{Fact}
\newtheorem{axiom}[theorem]{Axiom}
\newtheorem{example}[theorem]{Example}
\newtheorem{proposition}[theorem]{Proposition}
\newtheorem{remark}[theorem]{Remark}
\newtheorem{corollary}[theorem]{Corollary}
\newtheorem{lemma}[theorem]{Lemma}
\newtheorem{problem}[theorem]{Problem}
\newtheorem{note}[theorem]{Note}
\newtheorem{conjecture}[theorem]{Conjecture}
\newtheorem{assumption}[theorem]{Assumption}
\newtheorem{question}[theorem]{Question}
\newtheorem{property}[theorem]{Property}
\newcommand{\commentout}[1]{}
\newcommand{\PR}{\mathcal{PR}}
\newcommand{\wbox}{\hfill \mbox{$\sqcap$\llap{$\sqcup$}}}




\title{Translucent Players: Explaining 
Cooperative Behavior in Social Dilemmas}


\author{Valerio Capraro\\
Centre for Mathematics and Computer Science (CWI)\\
 Amsterdam, 1098 XG, The Netherlands\\
V.Capraro@cwi.nl
\and
Joseph Y. Halpern\\
   Cornell University\\
   Computer Science Department\\
   Ithaca, NY 14853\\
   halpern@cs.cornell.edu}
\maketitle

\begin{abstract}
In the last few decades, numerous experiments have shown that humans do not always behave so as to maximize their material payoff. Cooperative behavior when non-cooperation is a dominant strategy (with respect to the material payoffs) is particularly puzzling. Here we propose a novel approach to explain cooperation, assuming what Halpern and Pass \citeyear{HaPa13} call
\emph{translucent players}. Typically, players are assumed to be \emph{opaque}, in the sense
that a deviation by one player does not affect the strategies used by other players. But a player may believe that if he switches from one strategy to another, the fact that he chooses to switch may be visible to the other players.  For example, if he chooses to defect in Prisoner's Dilemma, the other player may sense his guilt.  We show that by assuming translucent players, we can recover many of the regularities observed in human behavior in well-studied games such as Prisoner's Dilemma, Traveler's Dilemma, Bertrand Competition, and the Public Goods game.
\end{abstract}






\section{Introduction}\label{sec:intro}

In the last few decades, numerous experiments have shown that humans do not
always behave so as to maximize their material 
payoff.  Many alternative models have consequently been proposed to
explain  deviations from the money-maximization paradigm. 
Some of them assume that players are boundedly rational 
and/or make mistakes in the computation of the expected utility of a strategy
\cite{CHC04,CG-Cr-Br01,HP08,MK-Pa95,St-Wi94}; yet others assume 
that players have other-regarding preferences \cite{Bo-Oc,Ch-Ra,Fe-Sc};
others define radically different solution concepts, assuming that
players do not try to maximize their payoff, but rather try to
minimize their regret \cite{HP11b,RS08}, or maximize the forecasts
associated to coalition structures \cite{Ca,CVPJ}, or maximize the total
welfare \cite{AptSchafer14,RH}.  (These references only scratch the surface;
a complete bibliography would be longer than this paper!)

\commentout{
 one more model of human behavior, taking its origin
from the influential Nature paper by Rand, Green, and Nowak
\cite{RGN}. In this paper and its subsequent works
\cite{R13a,R13b,R13c}, Rand and colleagues have provided evidence that
behavior of experimental subjects in the lab is strongly influenced by
their experience in everyday interactions. The Social Heuristics
Hypothesis affirms that people internalize strategies that are more
successful in everyday interactions and use them as default strategies
in the lab.  

In this light, it becomes important to formalise the kind of reasoning
people are likely to follow in a real-life scenario and try to develop a
different equilibrium theory which hopefully better explains the
experimental findings that have been collected in the last decades.  
}


Cooperative behaviour in one-shot anonymous games is particularly
puzzling, especially in games where 
non-cooperation is a dominant strategy (with
respect to the material payoffs): why should you pay a cost to help
a stranger, when no clear direct or indirect reward seems to be at stake?
Nevertheless, 
the secret of success of our societies is largely due to our ability
to cooperate. We do not cooperate only with family members, friends,
and co-workers. A great deal of cooperation can be observed also in
one-shot anonymous interactions \cite{Camerer03}, where none of the
five rules of cooperation proposed by Nowak \citeyear{Nowak06} seems
to be at play.   

Here we propose a novel approach to explain cooperation, based on work 
of Halpern and Pass \citeyear{HaPa13} 
and Salcedo \citeyear{Salcedo}, assuming what Halpern and Pass call
\emph{translucent players}.   Typically,
players are assumed to be \emph{opaque}, in the sense
that a deviation by one player does not affect the strategies used by
other players. But a player may believe that if he switches from one
strategy to another, the fact that he chooses to switch may be visible to
the other players.  For example, if he chooses to defect in 
Prisoner's Dilemma, the other player may sense his guilt.  (Indeed, it is well
known that there are facial and bodily clues, such as increased pupil size,
associated with deception; see, e.g., \cite{EkmanFriesen69}.
Professional poker players are also very sensitive to
\emph{tells}---betting patterns and physical demeanor that reveal
something about a player's hand and strategy.)

We use the idea of translucency to explain cooperation.  This
may at first seem somewhat strange. Typical lab experiments of social
dilemmas consider anonymous players, who play each other over computers.
In this setting, there are no tells.  
%
However, as Rand and his colleagues have argued (see, e.g.,
\shortcite{RGN,R13a}),  
behavior of subjects in lab experiments is strongly influenced by
their experience in everyday interactions.  People  internalize
strategies that are more successful in everyday interactions and use
them as default strategies in the lab.   We would 
argue that people do not just internalize 
strategies; they also internalize \emph{beliefs}.  In everyday
interactions, changing strategies certainly affects how other players react
in the future.  Through tells and leaks, it also may affect how other
players react in current play.  Thus, we would argue that in everyday
interactions, people assume a certain amount of translucency, both
because it is a way of taking the future into account in real-world
situations that are repeated and because it is a realistic assumption in
one-shot games that are played in settings where players have a great
deal of social interaction.  We claim that players then apply these
beliefs in lab settings where they are arguably inappropriate.

There is some additional experimental evidence that can be viewed as
supporting translucency.  There is growing evidence that showing
people simple images of watching eyes has a marked effect on behavior,
ranging from giving more in Public Goods games to littering less (see
\shortcite{BCHRN13} for a discussion of some of this work and an extensive
list of references).  
One way of understanding these results is that the eyes are making
people feel more translucent.

\commentout{
Halpern and Pass formalize translucency using what they call
\emph{counterfactual structures}.  In a counterfactual structure, each
player $i$ has beliefs about what other players would play if $i$ were
to deviate.  The key point is that, unlike standard frameworks, these
beliefs may change depending on how $i$ deviates.    Like standard
model, counterfactual structures involve a state space;  associated with
each state $\omega$ is the strategy profile $\strat(\omega)$ played at
$\omega$.  Players' beliefs regarding the effect of deviations is
captured using a function $f$ that associates with each state $\omega$,
player $i$, and strategy $\ssigma$ for player $i$ the state
$f(\omega,i,\ssigma')$ that  would result if $i$ deviated from the
strategy $\ssigma$ that he actually uses in state $\omega$ to $\ssigma'$.
(If $\ssigma = \ssigma'$, then $f(\omega,i,\ssigma) = \omega$.)  Each
player $i$ is assumed to have a belief about what other players are
doing (captured by a distribution over states).  The function $f$ induces a

distribution over the states that would result in $i$ switched to
$\ssigma'$, which in turn results in $i$ having beliefs about what the
other players would do if he were to switch to $\ssigma'$.

Halpern and Pass \citeyear{Ha-Pa13} define a solution concept in
counterfactual structures that they call \emph{iterated minimax
domination}, which can be viewed as a generalization of
rationalizability \cite{Ber84,Pearce84}.%
\footnote{Iterated minimax deletion was also defined independently by
Capraro \citeyear{Capraro13} and Salcedo \citeyear{Salcedo}.}
Here we define another
solution concept in counterfactual structures that we call
\emph{translucent equilibrium}; it is arguably more in the 

spirit of Nash equilibrium.  Not surprisingly,
every Nash equilibrium is a translucent equilibrium in a model where
each player $i$ believes that other players' strategies do not change if he
deviates (so that if $\omega' = f(\omega,i,\ssigma')$, then for all
players $j \ne i$ use the same strategy in $\omega$ and $\omega'$).
}

\commentout{
A major difference between idealised strategic situations and real-life
strategic situations is \emph{opacity}. The classical Nash equilibrium
notion is stable under unilateral deviations, which means essentially
that the reasoning of one agent cannot influence the reasoning of the
others: if player $i$ reasons about changing his strategy from $s_i$ to
$s_i'$ in strategy profile $(s_i,s_{-i})$, he assumes that the other
players do not change their strategy and keep staying in $s_{-i}$. On
the other hand, real-life interactions are shaped in a context where
agents are in contact with each other and they can consequently
influence each other by several ways of communication: if player $i$
reasons about changing his strategy from $s_i$ to $s_i'$, the other
players may \emph{see} this reasoning and change their strategy as
well. In other words, players are translucent. 

Using counterfactual structures \cite{St68,Le73}, Halpern and Pass
\cite{Ha-Pa13} made the first step towards this direction and defined
games with translucent players. For these games they defined a notion of
rationality and characterized rational strategies. 
}

\commentout{
Translucent equilibria are possibly mixed strateg profiles; by way of
contrast, when considering the strategies that survive iterated minimax
deletion, the focus is on pure strategies, 
 on strategies of individuals players, rather than strategy profiles.
We take the probabilities in mixed equilibria to characterize the
proportion of subjects that will play a given strategy.  We can then
compare these predicted proportions to what is observed in the lab.
Another key feature of translucent equilibria is that they are defined
relative to a counterfactual structure.  
Counterfactual structures are very flexible.  We  can capture many
situations, depending on how the closest-state function $f$ is defined.
However, this very flexibility takes away from the predictive power of
counterfactual structures and translucent equilibrium.  As a first step
to recovering some predictive power, we focus on
}

We apply the idea of translucency to 
a particular class of games that we call \emph{social
dilemmas} (cf.~\cite{Dawes80}).
  A social dilemma is a normal-form game with two properties:
\begin{enumerate}
\item there is a unique Nash equilibrium $\ssigma^N$, which is a pure
strategy profile;
\item there is a unique welfare-maximizing profile $\ssigma^W$,
again a pure strategy profile, 
such that each player's utility if $\ssigma^W$ is played is higher
than his utility if $\ssigma^N$ is played.
\end{enumerate}
Although social dilemmas are clearly a restricted class of games,
they contain some of the best-studied games in the game theory
literature, including Prisoner's Dilemma, Traveler's Dilemma
\cite{Basu94}, Bertrand Competition, and the Public Goods game.
(See Section~\ref{se:social dilemmas model} for more discussion of these games.)

There are (at least) two reasons why an agent may be concerned about
translucency in a social dilemma: (1) his opponents may discover
that he is planning to defect and punish him by defecting as well, (2)
many other people in his social group (which may or may not include
his opponent) may discover that he 
is planning to defect (or has defected, despite the fact that the game
is anonymous) and think worse of him.  

 
For definiteness, we focus here on the first point and assume that, in social
dilemmas, players have a degree of belief 
$\alpha$ that they are translucent, so that if they intend to
cooperate (by playing their component of the welfare-maximizing
strategy) and decide to deviate, there is a probability $\alpha$
that another player will detect this, and play her component of
the Nash equilibrium strategy.  (These detections are independent, so
that the probability of, for example, exactly two players 
other than $i$ detecting a deviation by $i$ is 
$\alpha^2(1-\alpha)^{N-3}$, where $N$ is the total number
of players.)
Of course, if $\alpha = 0$, then we are back at the standard
game-theoretic framework.  We show that, with this assumption, we can
already explain
a number of 
experimental regularities observed in social dilemmas (see 
Section~\ref{se:social dilemmas model}).
We can model the second point regarding concerns about translucency
in much the same way, and would get
qualitatively similar results (see Section~\ref{sec:discussion}).

The rest  of the paper is as follows. 
In Section~\ref{se:definitions}, we formalize the notion of
translucency in a game-theoretic setting.
In Section \ref{se:social dilemmas model}, we define the social
dilemmas that we focus on in this paper; in 
Section \ref{sec:explanation}, we show that
by assuming translucency, we can obtain as predictions of the
framework a number of regularities that have been observed in the
experimental literature. 
In Section~\ref{sec:discussion}, we show that most of the other
approaches proposed for  
explaining human behavior in social dilemmas do not predict all 
these regularities.

In the appendix, we discuss a solution concept that we call
\emph{translucent equilibrium}, based on translucency, closely related
to the notion of \emph{individual rationality} discussed by Halpern and
Pass \citeyear{HaPa13}, and show how it can be applied in social dilemmas.

\section{Rationality with translucent players}\label{se:definitions}

\commentout{
In this section, we briefly review the definition of counterfactual
structures.  The reader is encouraged to consult \cite{Ha-Pa13} for more
detail and intuition.}
In this section, we briefly define rationality in the presence of
translucency, motivated by the ideas in Halpern and Pass \citeyear{HaPa13}.

Formally, a (finite) normal-form  game $\mathcal G$ is a tuple
$(P,\SSigma_1,\ldots,\SSigma_N,u_1,\ldots,u_N)$,
where $P=\{1,\ldots,N\}$ is the set of players, $\SSigma_i$
is the set of  strategies for player $i$,  and $u_i$ is player $i$'s
utility function.    Let $\SSigma = \SSigma_1 \times \cdots \times
\SSigma_N$ and $\SSigma_{-i} = \prod_{j \ne
  i} \SSigma_j$.    We assume that $\SSigma$ is finite
and that $N \ge 2$.  

In standard game theory, it is assumed that a player $i$ has beliefs about 
the strategies being used by other players; $i$ is rational if his
strategy is a best response to these beliefs.  The standard definition
of best response is the following.

\begin{definition}\label{brstandard}
A strategy $s_i \in \SSigma_i$ is a best
response to a probability $\mu_i$ on $\SSigma_{-i}$ if, for all strategies $s'_i$ for player $i$, we have $$\sum_{s'_{-i} \in \SSigma_{-i}} \mu_i(s_{-i}') u_i(s_i,s'_{-i}) \ge 
\sum_{s'_{-i} \in \SSigma_{-i}} \mu_i(s'_{-i}) u_i(s'_i,s'_{-i}).$$
 \end{definition}
Definition~\ref{brstandard} implicitly  assumes that $i$'s beliefs about
what other agents are doing do not change if $i$ switches from $s_i$,
the strategy he was intending to play,
to a different strategy.  (In general, we assume that $i$ always has 
an \emph{intended strategy}, for otherwise it does not make sense
to talk about $i$ switching to a different strategy.)  
So what we really have are beliefs $\mu_i^{s_i,s_i'}$ for $i$ indexed by a pair
of strategies $s_i$ 
and $s_i'$; we interpret $\mu_i^{s_i,s_i'}$ as $i$'s beliefs if he intends
to play $s_i$ but instead deviates to $s_i'$.  Thus, $\mu_i^{s_i,s_i}$
represents $i$'s beliefs if he plays $s_i$ and does not deviate.  

We can now define a best response for $i$ with respect to a family of
beliefs  $\mu_i^{s_i,s_i'}$.  

\begin{definition}\label{brdefinitionnew}  Strategy $s_i \in \SSigma_i$ is a \emph{best response
    for $i$ with respect to the beliefs $\{\mu_i^{s_i,s_i'}: s_i' \in
    \SSigma_i\}$} if, for all strategies $s_i' \in \SSigma_i$, we have 
$$\sum_{s'_{-i} \in \SSigma_{-i}} \mu_i^{s_i,s_i}(s_{-i}') u_i(s_i,s'_{-i}) \ge 
\sum_{s'_{-i} \in \SSigma_{-i}} \mu_i^{s_i,s_i'}(s'_{-i}) u_i(s'_i,s'_{-i}).$$

\end{definition}

We are interested in players
who are making best responses to their beliefs, but we define best
response in terms of Definition~\ref{brdefinitionnew}, not
Definition~\ref{brstandard}.  
Of course, the standard notion of best response is
just the special case of the notion above where 
$\mu_i^{s_i,s_i'} = \mu_i^{s_i,s_i}$ for all $s_i'$: a player's beliefs about what other
players are doing does not change if he switches strategies.

\begin{definition}
We say that a player is \emph{translucently rational} if he best
responds to his beliefs in the sense of Definition
\ref{brdefinitionnew}. 
\end{definition}

Our assumptions about translucency will be used to determine 
$\mu_i^{s_i,s_i'}$.  For example, suppose that $\Gamma$ is a 2-player
game, player $1$ believes
that, if he were to switch from $s_i$ to $s_i'$, this would be detected by
player 2 with probability $\alpha$, and if 
player $2$ did detect the switch, then player $2$ would switch to 
$s_j'$.  Then $\mu_i^{s_i,s_i'}$ is $(1-\alpha)\mu^{s_i,s_i} +
\alpha\mu'$, where $\mu'$ assigns probability 1 to
$s_j'$; that is, player 1 believes that with probability $1-\alpha$,
player 2 continues to do what he would have done all along (as
described by $\mu^{s_i,s_i}$) and, with probability $\alpha$, player 2
switches to $s_j'$.

\commentout{
\begin{definition}\label{defin:counterfactual structure}
{\rm A finite counterfactual structure appropriate for the game
$\mathcal G$ is a tuple $M = (\Omega, \mathbf{s}, f, \PR_1,\ldots,\PR_N)$,
where: 
\begin{itemize}
\item $\Omega$ is a finite space of states;
\item $\mathbf{s}:\Omega\to \SSigma$
\item For each $\omega\in\Omega$, $\mathcal P\mathcal R_i(\omega)$ is a
probability measure on $\Omega$ satisfy the following properties:
\begin{itemize}
\item[PR1.] $\PR_i(\omega)(\{\omega'\in\Omega :
\mathbf{s}_i(\omega')=\mathbf{s}_i(\omega)\})=1$ (where
$\strat_i(\omega)$ denotes player $i$'s strategy in $\strat(\omega)$;
\item[PR2.] $\PR_i(\omega)(\{\omega'\in\Omega :
\PR_i(\omega')=\PR_i(\omega)\})=1$.
\end{itemize}
These assumptions guarantee that player $i$ assigns probability $1$ to
his actual strategy and beliefs. 
\item $f$ is the closest-state function: given $\omega\in\Omega$, $i\in
P$, $\ssigma'\in \SSigma_i$, $f(\omega,i,\ssigma')\in\Omega$ represents what would
happen if player $i$ switched strategy to $s_i'$ at state $\omega$. The
closest-state function $f$ is required to satisfy the following
properties:
\begin{itemize}
\item[CS1.] $\mathbf{s}_i(f(\omega,i,\ssigma'))=\ssigma_i'$ (so that at
$f(\omega,i,\ssigma_i')$, player $i$ plays $\ssigma_i'$);
\item[CS2.] $f(\omega,i,\mathbf{s}_i(\omega))=\omega$.   \wbox
\end{itemize}
\end{itemize}
}
\end{definition}


$\PR_i(\omega)$ defines $i$'s beliefs at $\omega$.  We now define
$\PR_{i,\ssigma'}(\omega)$: $i$'s beliefs at $\omega$ if he were to
switch to strategy $\ssigma'$ by taking
$$\mathcal P\mathcal R_{i,\ssigma'}(\omega)(\omega')=\sum_{\{\omega'' :
f(\omega'',i,\ssigma')=\omega'\}}\mathcal P\mathcal R_i(\omega)(\omega''). 
$$
Intuitively, if he were to switch from to strategy $\ssigma'$ at
$\omega$, the probability that $i$ would assign to state $\omega'$ is
the sum of the probabilities that he assigns to all the states
$\omega''$ such that he beliefs that he would move from $\omega''$ to
$\omega'$ if he used strategy $\ssigma'$.  This is exactly sum of 
$\Pr(\omega)(\omega'')$ for all $\omega''$ such that
$f(\omega'',i,\ssigma') = \omega'$.

We define the expected utility of player $i$ at state $\omega$ in the
usual way as the sum of the product of his expected utility of the
strategy profile played at each state $\omega'$ and the
probability of $\omega'$:
$\EU_i(\omega)=\sum_{\omega'\in\Omega}\mathcal P\mathcal
R_i(\omega)(\omega')u_i(\mathbf{s}_i(\omega),\mathbf{s}_{-i}(\omega')).$%
\protect{\footnote{Given a profile $t = (t_1, \ldots, t_N)$, as usual, we define
$t_{-i}=(t_1,\ldots,t_{i-1},t_{i+1},\ldots,t_N)$.  We extend this
notation in the obvious way to functions like $\strat$, so that, for
example, 
$\strat_{-i}(\omega) = (\strat_1(\omega), \ldots, \strat_{i-1}(\omega),
\strat_{i+1}(\omega),\ldots, \strat_{n}(\omega))$.}

The usual way of defining $i$'s expected utility at $\omega$ if he were
to switch to $\ssigma'$ is to simply replace his actual strategy at
$\omega$ by $\ssigma'$ at all states, keeping the strategies of the other
players the same; that is,
$$\sum_{\omega'\in\Omega}\PR_i(\omega)(\omega')u_i(\ssigma',\mathbf{s}_{-i}(\omega')).$$
In this definition,  player $i$'s beliefs about the strategies that the
other players are using does not change when he switches from
$\strat_i(\omega)$ to $\ssigma'$.  The key point of counterfactual
structures is that this belief may well change. Thus, we define $i$'s
expected utility at $\omega$ if he switches to $\ssigma'$ as
$$
\EU_i(\omega,\ssigma')=\sum_{\omega'\in\Omega}\mathcal
\PR_{i,\ssigma'}(\omega)(\omega')u_i(\ssigma',\mathbf{s}_{-i}(\omega')).  
$$

Finally, we can define rationality in counterfactual structures using
these notions: 
\begin{definition}\label{def:rationality}
{\rm Player $i$ is \emph{rational} at state $\omega$ if, for all
$\ssigma' \in \SSigma_i$,
$$\EU_i(\omega) \ge \EU_i(\omega,\ssigma').$$
\wbox
}
\end{definition}

\section{Translucent equilibrium}\label{se:equilibria}
In this section, we define translucent equilibrium.  We start with some
preliminary definitions.

Given a probability measure $\tau$ on a finite set $T$, let
$\text{supp}(\tau)$ denote the support of $\tau$, that is,  
$$
\text{supp}(\tau)=\{t\in T : \tau(t)\neq0\}.
$$

Note that $\sigma_{-i}$ can 
can be viewed as a probability on
$S_{-i}$, where
$\sigma_{-i}(s_{-i})=\prod_{j\neq i}\sigma_{j}(s_j)$.
Similarly $\sigma$ can be viewed as a probability measure on $S$.
In the sequel, we view $\sigma_{-i}$ and $\sigma$ as probability
measures without further comment (and so talk about their support).

\begin{definition}\label{def:equilibrium}
{\rm A strategy profile $\sigma$ in a game $\G$ is 
\emph{translucent equilibrium} in a counterfactual structure $M =
(\Omega, \mathbf{s}, f, \PR_1,\ldots,\PR_N)$ 
appropriate for $\G$ if there exists a subset $\Omega' \subseteq \Omega$
such that, for each state $\omega$ in $\Omega'$, the following
properties hold:
\begin{enumerate}
\item[TE1.] $\mathbf{s}(\omega)\in\text{supp}(\sigma)$.
\item[TE2.] $\text{supp}(\PR_i(\omega))\subseteq\Omega'$.
\item[TE3.] $\mathbf{s}_{-i}(\PR_i(\omega))=\sigma_{-i}$ (i.e.,
for each strategy profile $s_{-i} \in S_{-i}$, we have 
$\sigma_{-i}(s_{-i}) = \PR_i(\omega)(\{\omega': \strat_{-i}(\omega') =
s_{-i}\})$). 
\item[TE4.] Each player is rational at $\omega$.
\end{enumerate}
The mixed strategy profile $\sigma$ is a translucent equilibrium of $\G$
if there exists a counterfactual structure $M$ appropriate for $\G$ such
that $\sigma$ is a translucent equilibrium in $M$.}  \wbox
\end{definition}
Intuitively, $\sigma$ is a translucent equilibrium in $M$ if, for each
strategy $s_i$ in the support of $\sigma_i$, the expected utility of
playing $s_i$ given that other players are playing according
$\sigma_{-i}$ is at least as good as switching to some other strategy
$s_i'$, given what $i$ would believe about what strategies the other
players are playing if he were to switch to $s_i'$.  

It is easy to see that every Nash equilibrium of $\G$ is a translucent
equilibrium.  

\begin{proposition} Every Nash equilibrium of $\G$ is a translucent
equilibrium.
\end{proposition}
\begin{proof} 
Given a Nash equilibrium 
$\sigma=(\sigma_1,\ldots,\sigma_n)$, consider the following
counterfactual structure $M_\sigma = (\Omega, \mathbf{s}, f,
\PR_1,\ldots,\PR_N)$:
\begin{itemize}

\item $\Omega$ is the set of strategy profiles in the support of
$\sigma$;
\item $\strat(s) = s$;
\item $\PR_i(s_i,s_{-i})(s_i',s_{-i}') = 
\left\{
\begin{array}{ll}
0 &\mbox{if $s_i' \ne s_i$}\\
\sigma_{-i}(s_{-i}') &\mbox{if $s_i' = s_i$;}
\end{array}
\right.$
\item $f((s_i,s_{-i}),i,s')=(s',s_{-i})$. 
\end{itemize}

It is easy to check that $\sigma$ is a translucent equilibrium in
$M_\sigma$; we simply take $\Omega' = \Omega$.  The fact that $f$ is an
``opaque'' closest-state function, which is not affected by the strategy
used by players, means that rationality in $M$ reduces to the standard
definition of rationality.  We leave details to the reader.  
\end{proof}

Although the fact that we can consider arbitrary counterfactual
structures (appropriate for $\G$) means that many strategy profiles are
translucent equilibria, the notion of translucent equilibrium has some
bite.  For example, the strategy profile $(C,D)$, where player 1
cooperates and player 2 defects, is not a translucent
equilibrium in Prisoner's Dilemma.  If player 1 believes that player 2
is playing defecting with probability 1, there are no beliefs that 1
could have that would justify cooperation.  However, as we shall see,
both $(C,C)$ and $(D,D)$ are translucent equilibria.  This follows from
the characterization of translucent equilibrium that we now give.

\begin{definition}\label{dfn:coherent}
{\rm Given a game $\G = (P,S_1,\ldots, S_N, u_1, \ldots, u_N)$, 
a mixed-strategy profile $\sigma$ in $\G$ is \emph{coherent} if
for all players $i \in P$, all $s_i\in\text{supp}(\sigma_i)$, and all
$s_i'\in S_i$, there is $s_{-i}'\in S_{-i}$ such that 
$$u_i(s_i,\sigma_{-i})\geq u_i(s')$$
(where, of course, $u_i(s_i,\sigma_{-i}) = \sum_{s_{-i}'' \in S_{-i}'}
\sigma_{-i}(s_{-i}'') u_i(s_i, s_{-i}'')$).
} \wbox
\end{definition}
That is, $\sigma$ is coherent if, for all pure strategies for player $i$
in the support of $\sigma_i$, if $i$'s belief about the strategies being
played by the other  players is given by $\sigma_{-i}$, 
strategies is given by there is no obviously better strategy that
$i$ can switch to in the weak sense that, if $i$ contemplates switching
to $s_i'$, there are beliefs that $i$ could have about the other players
(namely, that they would definitely play $s_{-i}'$ in this case) that
would make switching not make to $s_i'$ better than sticking with $s_i$.

It is easy to see that $(C,C)$ and $(D,D)$ in Prisoner's Dilemma are
both coherent;  on the other hand, $(C,D)$ is not.  

\begin{theorem}\label{thm:translucenteq}
The mixed strategy profile $\sigma$ of game $\G$ is coherent iff $\sigma$ is a
translucent equilibrium of $\G$.n
\end{theorem}

\begin{proof}
Fix a game $\G = (\Omega, \mathbf{s}, f, \PR_1,\ldots,\PR_N)$, and
suppose that
$\sigma$ is a coherent strategy profile in $\G$.  We construct a counterfactual
structure $M = (\Omega, \mathbf{s}, f, \PR_1,\ldots,\PR_N)$ as follows:
\begin{itemize}
\item $\Omega = S$;
\item $\strat(s) = s$;
\item $\PR_i(\omega)(\omega') = \left\{
\begin{array}{ll}
1 &\mbox{if $\omega \notin \supp(\sigma_i)$, $\omega = \omega'$}\\
0 &\mbox{if $\omega \notin \supp(\sigma_i)$, $\omega \ne \omega'$}\\
\sigma_{-i}(\strat_{-i}(\omega')) &\mbox{if $\omega \in \supp(\sigma_i)$,
$\strat_i(\omega') = \strat_i(\omega)$}\\
0 &\mbox{if $\omega \in \supp(\sigma_i)$,
$\strat_i(\omega') \ne \strat_i(\omega)$};
\end{array}
\right.$
\item $f(\omega,i,s_i') = \left\{
\begin{array}{ll}
(s_i',\mathbf{s}_{-i}(\omega)) &\mbox{if $\omega \notin
\supp(\sigma_i)$}\\
\omega &\mbox{if $\omega \in \supp(\sigma_i)$, $s_i' = \strat_i(\omega)$}\\
(s_i',s_{-i}') & \mbox{if $\omega \in \supp(\sigma_i)$, $s_i' \ne
\strat_i(\omega)$, where $s_{i}'$ is a}\\ &\mbox{strategy such that
$u_i(\mathbf{s}_i(\omega),\sigma_{-i})\geq u_i(s')$;} \\
&\mbox{such a strategy is
guaranteed to exist since}\\ &\mbox{$\sigma$ is coherent.} 
\end{array}
\right.$
\end{itemize}

We first show that $M$ is a finite counterfactual structure appropriate
for $\G$; in particular, $\PR_i$ satisfies PR1 and PR2 and $f_i$
satisfies CS1 and CS2.  For PR1 and PR2, there are two cases.  If
$\omega \notin \supp(\sigma)$, then $\PR_i(\omega)(\omega) = 1$, so PR1
and PR2 clearly hold.  If $\omega \notin \supp(\omega)$, then 
$\PR_i(\omega)(\omega) > 0$ iff $\strat_i(\omega) = \strat_i(\omega')$.
Moreover, if $\strat_i(\omega) = \strat_i(\omega')$, then it is
immediate from the definition that $\PR_i(\omega) = \PR_i(\omega')$, so
PR2. holds.  That CS1 and CS2 hold is immediate from the definition of $f_i$.

To show that $\sigma$ is a translucent equilibrium in $M$, let
$\Omega' = \supp(\sigma)$.  For each state $\omega \in \Omega'$, TE1
clearly holds.  Note that if $\omega \in
\supp(\sigma)$, then $\PR_i(\omega) = (\strat_i(\omega),
\sigma_{-i}(\omega))$ (identifying the strategy profile with a
probability measure), so TE2 and TE3 clearly hold. 
It remains to show that TE4 holds, that is, 
that every player is rational at every state $\omega\in\Omega'$.

Thus, we must show that $\EU_i(\omega) \ge \EU(\omega,s_i^*)$ for all
$s_i^* \in S_i$.  Note that
$$\begin{array}{lll}
\EU_i(\omega)
&= &\sum_{\omega'\in\Omega}\mathcal
\PR_i(\omega)(\omega')u_i(\mathbf{s}_i(\omega),\mathbf{s}_{-i}(\omega')) \\
&= &\sum_{\{\omega'\in\Omega: \strat_i(\omega') = \strat_i(\omega)\}}
\sigma_{-i}(\strat_{-i}(\omega'))
u_i(\mathbf{s}_i(\omega),\mathbf{s}_{-i}(\omega'))\\ 
&= &\sum_{s_{-i}''\in S_{-i}}
u_i(\mathbf{s}_i(\omega),s_{-i}'')\\ 
&=&u_i(\strat_i(\omega),\sigma_{-i}).
\end{array}
$$
By definition,
$$
\EU_i(\omega,\ssigma_i^*)=\sum_{\omega'\in\Omega}\mathcal
\PR_{i,\ssigma_i^*}(\omega)(\omega')u_i(\ssigma_i^*,\mathbf{s}_{-i}(\omega'))$$
and 
$$\PR_{i,\ssigma'}(\omega)(\omega')=\sum_{\{\omega'' :
f(\omega'',i,\ssigma')=\omega'\}}\mathcal P\mathcal R_i(\omega)(\omega''). 
$$
Now if 
$s_i^*=\mathbf{s}_i(\omega)$, then $f(\omega,i,s_i^*)$.  In this case,
it is easy to check that $\PR_{i,\ssigma_i^*}(\omega) = \PR_i(\omega)$,
so 
$\EU_i(\omega,s_i^*)=\EU_i(\omega) = \EU_i(s_i,\sigma_{-i})$, and TE4
clearly holds.  
On the other hand, if 
$s_i^*\neq\mathbf{s}_i(\omega)$, then 
$$\begin{array}{lll}
\EU_i(\omega,s_i^*)
&=&\sum_{\omega'\in\Omega}\sum_{\{\omega'':f(\omega'',i,s_i^*)=\omega'\}}\PR_i(\omega)(\omega'')u_i(s_i^*,\mathbf{s}_{-i}(\omega'))\\
&=&\sum_{\{\omega'\in\Omega: \strat_i(\omega') =
s_i^*\}}\sum_{\{\omega'':f(\omega'',i,s_i^*)=\omega',\, 
\strat_i(\omega'') = \strat_i(\omega)\}}\sigma_{-i}(\omega'')
u_i(s_i^*,\mathbf{s}_{-i}(\omega'))\\ 
&=&\sum_{\{\omega'\in\Omega: \strat_i(\omega') =
s_i^*\}}\sum_{\{\omega'':f(\omega'',i,s_i^*)=\omega',\, 
\strat_i(\omega'') = \strat_i(\omega)\}}\sigma_{-i}(\omega'')
u_i(f(\omega'',i,s_i^*)).\\ 
\end{array}
$$
By definition, $u_i(f(\omega'',i,s_i^*))\leq
u_i(\strat_i(\omega''),\sigma_{-i})=u_i(\strat_i(\omega),\sigma_{-i})$. Thus,
$$\begin{array}{lll}
\EU_i(\omega,s_i^*)
&\le&\sum_{\{\omega'\in\Omega: \strat_i(\omega') =
s_i^*\}}\sum_{\{\omega'':f(\omega'',i,s_i^*)=\omega',\, 
\strat_i(\omega'')= \strat_i(\omega)\}}\sigma_{-i}(\omega'')
u_i(\strat_i(\omega),\sigma_{-i})\\ 
&=&u_i(\strat_i(\omega),\sigma_{-i}) \sum_{\{\omega'\in\Omega: \strat_i(\omega') =
s_i^*\}}\sum_{\{\omega'':f(\omega'',i,s_i^*)=\omega',\, 
\strat_i(\omega'') = \strat_i(\omega)\}}\sigma_{-i}(\omega'')\\
&=&u_i(\strat_i(\omega),\sigma_{-i}).
\end{array}
$$
This completes the proof that TE4 holds, and the proof of the ``only if''
direction of the argument

The ``if'' is actually much simpler. Suppose, by way of contradiction,
that $\sigma$ is not coherent.  Then
there is a player $i$ and a strategy $s_i\in\text{supp}(\sigma_i)$ such
that for all $s_{-i}' \in S_i$, we have $u_i(s_i,\sigma_{-i})<u_i(s')$. It
follows that, for all counterfactual structures $M$, no matter what the
beliefs and the closest-state functions are in $M$,  
it is always strictly profitable for player $i$ to switch strategy from
$s_i$ to $s_i'$. Consequently, $i$ is not rational at a state $\omega$
such that $s_i(\omega)=s_i$, contradicting TE4.
\end{proof}

We conclude by comparing translucent equilibrium to iterated minimax
domination.  We begin by reviewing the relevant definitions.

\begin{definition} {\rm Strategy $s_i$ for player $i$ in game $\G
= (P,\SSigma_1,\ldots,\SSigma_N,u_1,\ldots,u_N)$ is 
\emph{minimax dominated with respect to $S'_{-i}\subseteq S_{-i}$} iff
there 
exists a strategy $s'_i \in S_i$ such that  
$$\min_{t_{-i} \in
S'_{-i}} u_i(s'_i, t_{-i}) > \max_{t_{-i} \in 
S'_{-i}} u_i(s_i, t_{-i}).$$}
\wbox
\end{definition}

\noindent Thus, player $i$'s strategy $s_i$ is minimax dominated
with respect to $S'_{-i}$ iff there exists
a strategy $s_i'$ for player $i$ such that the worst-case payoff for
player $i$ if he uses $s'$ is strictly better than his best-case
payoff if he uses $s$, given that the other players are restricted to using a
strategy in $S'_{-i}$.

\begin{definition} Given a game $\G =
(P,\SSigma_1,\ldots,\SSigma_N,u_1,\ldots,u_N)$,  
define $\NSD_j^k(\G)$ inductively: let $\NSD_j^0(\G) =
S_j$ and let $\NSD_j^{k+1}(\G)$
consist of the strategies in $\NSD_j^{k}(\G)$ not minimax 
dominated with respect to $\NSD_{-j}^{k}(\G)$.
Strategy $s \in S_i$ \emph{survives iterated minimax domination}
if $s \in \cap_k\NSD_i^k(\G)$.  \wbox
\end{definition}

As these definitions show, only pure strategies of individual players
survive (or do not survive) iterated minimax domination.  It is not hard
to show that both $C$ and $D$ survive iterated minimax domination.  But,
as we have observed, although $(C,C)$ and $(D,D)$ are translucent
equilibria in Prisoner's Dilemma, $(C,D)$ is not.  So not every profile
made up of strategies that survive iterated minimax domination is a
translucent equilibrium.  On the other hand, as the following example
shows, the strategies in a profile that is a translucent equilibrium may
not survive iterated minimax domination.  

\begin{example} {\rm Consider the two-player game where $S_1 =
\{s_1^1,s_2^1,s_3^1\}$, $S_2 = \{s_1^2,s_2^2,s_3^2\}$, and the utilities
are defined by the following table:
\begin{table}[htb]
\begin{center}
\begin{tabular}{|| c | c | c | c ||}
\hline
{} & $s_1^2$ & $s_2^2$ & $s_3^2$\\
\hline
$s_1^1$ & (1,1) & (1,2) & (1,3) \\
\hline
$s_2^1$ & (2,1) & (2,2) & (2,3) \\
\hline
$s_3^1$ & (3,1) & (3,2) & (3,3)) \\
\hline
\end{tabular}
\end{center}
\end{table}
Thus, if player $i$ plays $s_i^k$, then his utility is $k$, independent
of what the other player does.

It easily follows from Theorem~\ref{thm:translucenteq} that
$(s_1^2,s_2^2)$ is a translucent equilibrium, since it is coherent (we
can take $s'$ in Definition~\ref{dfn:coherent} to be $(s_1^1,s_1^2)$.
On the other hand, $s_1^j$ and $s_2^j$ are minimax dominated by $s_3^j$.
\wbox }
\end{example}
}

}

\section{Social dilemmas}\label{se:social dilemmas model}

Social dilemmas are situations in which there is a tension
between the collective interest and individual interests: every
individual has an 
incentive to deviate from the common good and act selfishly, but if
everyone deviates, then they are all worse off. Personal and
professional relationships, depletion of natural resources, climate
protection, security of energy supply, and price competition in markets
are all instances of social dilemmas. 

As we said in the introduction, we formally define a social dilemma
as a normal-form game with a unique Nash equilibrium and a unique
welfare-maximizing profile, both pure strategy profiles, 
such that each player's utility if $\ssigma^W$ is played is higher
than his utility if $\ssigma^N$ is played.
While this is
a quite restricted set of games, it includes many that have been quite
well studied.  Here, we focus on the following games:
\begin{description}
\item[\textbf{Prisoner's Dilemma.}]  Two players can either cooperate
($C$) or defect ($D$).  To relate our results to experimental results on
Prisoner's Dilemma, we think of cooperation as meaning that a player
pays a cost $c > 0$ to give a benefit $b>c$ to the other player.  If a
player defects, he pays nothing 
and gives nothing.  Thus, the payoff of $(D,D)$ is
$(0,0)$, the payoff of $(C,C)$ is $(b-c,b-c)$, and the payoffs of $(D,C)$ and
$(C,D)$ are $(b,-c)$ and $(-c,b)$, respectively.  
Condition $b>c$ implies that $(D,D)$ is the unique Nash equilibrium
and $(C,C)$ is the unique welfare-maximizing profile.  

\item[\textbf{Public Goods game.}] $N\geq2$ contributors are endowed
with 1 dollar  each; they must simultaneously decide how much, if
anything, to contribute to a public pool. 
(The contributions must be in whole cent amounts.)
The total amount in the pot is then multiplied by a
constant strictly between 1 and $N$,
and then evenly redistributed among all players. So
the payoff of player $i$ is
$u_i(x_1,\ldots,x_N)=1-x_i+\rho(x_1+\ldots+x_N)$, where 
$x_i$ denotes $i$'s contribution, and $\rho\in\left(\frac1N,1\right)$
is the \emph{marginal return}.  
%
(Thus, the pool is multiplied by $\rho N$ before being split
evenly among all players.)  Everyone contributing nothing to
the pool is the unique Nash equilibrium, and everyone contributing their
whole endowment to the pool is the unique welfare-maximizing profile.

\item[\textbf{Bertrand Competition.}] $N\geq2$ firms compete to sell
their identical product  at a price between the ``price floor'' $L\geq 2$
and the ``reservation value'' $H$.
(Again, we assume that $H$ and $L$ are integers, and all prices must
be integers.)
 The firm that chooses the lowest
price, say $s$, sells the product at that price, getting a payoff of
$s$, while all other firms get a payoff of 0. If there are ties,  
then the sales are split equally among all firms that choose the lowest
price.   Now everyone choosing $L$ is the unique Nash equilibrium, and
everyone choosing $H$ is the unique welfare-maximizing profile.%
\footnote{We require that $L \ge 2$ for otherwise we would not have a
  unique Nash equilibrium, a condition we imposed on Social Dilemmas.
If $L = 1$ and $N=2$, we get two Nash
  equilibria: $(2,2)$ and $(1,1)$; similarly, for $L=0$, we also get
  multiple Nash equilibria, for all values of $N \ge 2$.}

\item[\textbf{Traveler's Dilemma.}] Two travelers have identical
luggage, which is damaged (in an identical way) by an airline.  The
airline offers to recompense them for their luggage. 
They may ask for any dollar amount between $L$ and $H$
(where $L$ and $H$ are both positive integers).
There is only one catch.  If they ask for the same amount, then that is
what they will both receive.  However, if they ask for different
amounts---say one asks 
for $m$ and the other for $m'$, with $m < m'$---then whoever asks
for $m$ (the lower amount) will get $m+b$ ($m$ and a bonus of $b$),
while the other player gets $m-b$: the lower amount and a penalty of
$b$.  It is easy to see that $(L,L)$ is 
the unique Nash equilibrium, while $(H,H)$ maximizes social welfare,
independent of $b$.
\end{description}


From here on, we say that a player \emph{cooperates} if he plays
his part of the socially-welfare maximizing strategy profile and
\emph{defects} if he plays his part of the Nash equilibrium strategy profile.

While Nash equilibrium predicts that people should always defect in
social dilemmas, in practice, we see a great deal of cooperative
behavior; that is, people often play (their part of) the
welfare-maximizing 
profile rather than (their part of) the Nash equilibrium profile.  Of
course, there have been many attempts to explain this.
Evolutionary theories may explain
cooperative behavior among genetically related individuals \cite{H64} or
when future interactions among the same subjects are
likely \cite{NS,T71}; see \cite{Nowak06} for a review of the five rules of cooperation. However, we
often observe cooperation even in
one-shot anonymous experiments 
among unrelated players
\cite{Rapoport}.


Although we do see a great deal of cooperation in these games, we do
not always see it.  
Here are some of the regularities that have been observed:
\begin{itemize}
\item The degree of cooperation in the Prisoner's dilemma depends
positively on the benefit of mutual cooperation and negatively on the
cost of cooperation \cite{CJR,EZ,Rapoport}.  
\item The degree of cooperation in the Traveler's Dilemma depends
negatively on the bonus/penalty \cite{CGGH99}.
\item The degree of cooperation in the Public Goods game depends
positively on the constant marginal return \cite{Gu,IWT}. 
\item The degree of cooperation in the Public Goods game depends
positively on the number of players \cite{BarceloCapraro,IWW,Ze}. 
\item The degree of cooperation in the Bertrand Competition depends
negatively on the number of players \cite{DG02}. 
\item The degree of cooperation in the Bertrand Competition depends
  negatively on the price floor \cite{D07}. 
\end{itemize}

\section{Explaining social dilemmas using translucency}\label{sec:explanation}

\commentout{
As we suggested in the introduction, we hope to use translucent
equilibrium to explain cooperation.  To do this, we need to construct
counterfactual structures that captures some of the intuition we have
about social welfare games.

\commentout{
We make the following hypotheses:

\begin{enumerate}
\item We consider $n$-player games such that:
\begin{itemize}
\item There is a unique welfare-maximizing profile of strategies $(w_1,\ldots,w_N)$;
\item There is a unique Nash equilibrium $(o_1,\ldots,o_N)$.
\end{itemize} 
\item Players are completely anonymous. This implies that they cannot
influence one another, that is, the reasoning of one player is
completely independent of the reasoning of the other players. 
\item Each player believes that the other players have a tension between
playing their optimal strategy and playing the welfare maximizing
strategy. 
\end{enumerate}
}

%
} 
 
\commentout{
Given a social dilemma $\G = (P,S_1,\ldots, S_N, u_1, \ldots, u_N)$, let 
$(s^N_1,\ldots, s^N_N)$ be the unique Nash equilibrium,
and let  $(s^W_1,\ldots, s^W_N)$ be the unique welfare-maximizing profile.
We now define a parameterized family of counterfactual structures
$M^\G(\sigma,\alpha) = 
(\Omega,\strat,\PR_1,\ldots,\PR_N)$, where $\alpha = 
(\alpha_1,\ldots, \alpha_N)$, $\alpha_i \in [0,1]$.
Intuitively, $\alpha_j$ describes how likely player $j$ is to
learn about a deviation (i.e., how ``translucent'' the other players
are to $j$) and $\sigma=(\sigma_{-i})_{i}$ describes what player $i$ believes the other players are going to do. We assume that if $j$ detects a deviation on the part of
any player, then she will play $s_j^N$. We also assume that the beliefs are supported on profiles of strategies $(s_1,\ldots,s_N)$, where $s_i\in\{s_i^N,s_i^W\}$. Moreover, we assume that $\sigma_{-i}$ is a power measure, in the sense that, for every $i,j$, with $i\ne j$, player $i$ believes that Player $j$ will be playing the strategy $\sigma_{i,j}:=p_{i,j}s_{j}^W+(1-p_{i,j})s_j^N$, where 
\begin{enumerate}
\item $p_{i,j}=:p_{-i}$ does not depend on $j$.
\item $\sigma_{-i}=\otimes_{j\ne i}\sigma_{i,j}$.
\end{enumerate} 
There are clearly a number of assumptions
being made here; we discuss these assumptions below.  We first define
$M^\G(\alpha, \PR')$ formally:
\begin{itemize}
\item $\Omega = S\times \{0,1\}^{n}$.  (The second component of the
state, which is an element of $\{0,1\}^{n}$, is used to determine the
closest-state function.  Roughly speaking, if $v_j = 1$, then player
$j$ learns  about a deviation if there is one; if $v_j = 0$, he does not.)
\item $\strat((s,v)) = s$.
\item $f((s,v),i,s_i^*) = 
\left\{
\begin{array}{ll}
(s,v) &\mbox{if $s_i = s_i^*$,}\\
(s',v) &\mbox{if $s_i \ne s_i^*$, where $s'_i = s_i^*$ and for $j
\ne i$,}\\
&\mbox{$s'_j = s_j$ if $v_j = 0$ and $s'_j = s^N_j$ if $v_j = 1$.}
\end{array}
\right.$

\noindent Thus, if player $i$ changes strategy from $s_i$ to $s_i'$, $s_i'\neq
 s_i$, then each other player $j$ either deviates to his component of
the Nash equilibrium or continues with his current strategy, 
depending on whether $v_j$ is 0 or 1. Roughly speaking, he switches to his
component of the Nash equilibrium if he learns about a deviation
(i.e., if $v_j = 1$).
\item $\PR_i(s,v)(s',v') =\left\{
\begin{array}{lll}
0 &\mbox{if $s_i \ne s'_i$ or $v_i \ne v_i'$,}\\
\sigma_{-i}(s'_i)\Pi_{\{j \ne i: v_j = 1\}}\alpha_j\Pi\, _{\{j \ne
i: \, v_j = 0\}}(1-\alpha_j) &\mbox{otherwise.}
\end{array}
\right.$
\noindent Thus, the probability of the $s'$ component of the state is
determined 
by $\sigma$ (with the standard assumption that player $i$ knows his
own strategy).  The probability of the $v'$ component is determined by
assuming that each player $j$ independently learns about a deviation
by $i$ with probability $\alpha_j$.\footnote{For simplicity, we are
  assuming that the probability that $j$ learns about a deviation by
  $i$ is the same for each player $i$.  More generally, we could have
  a probability $\alpha_{ji}$ for the probability that $j$ learns
  about a deviation by $i$.  Nothing significant would change in our analysis.}
\end{itemize}
\commentout{
\noindent Thus, if player $i$ changes strategy from $s_i$ to $s_i'$, $s_i'\neq
 s_i$, then each other player $j$ deviates to either his component of
the Nash equilibrium or his component of welfare-maximizing profile,
depending on whether $v_j$ is 0 or 1.
\commentout{
\textbf{Beliefs.}
To define the beliefs $\PR_i(\omega)$, fix state $\omega=(s,v)$ and
player $i$. As we mentioned earlier, the idea is that player $i$ assigns
probability $\alpha_i(j,\omega)\in[0,1]$ to the event that player $j$
($j\neq i$) changes from strategy $s_j$ to his or her Nash equilibrium
strategy $o_j$ and probability $1-\alpha_i(j,\omega)$ to the event that
player $j$ changes strategy from $s_j$ to his or her optimal
welfare-maximizing strategy $w_j$. We treat these probabilities as
parameters of the model. Starting from them, we can define
$\PR_i(\omega)$ as follows. Denote $R_i(s_i)$ the set of profiles of
strategies $s'$ such that $s_i'=s_i$ and $s_j'\in\{o_j,w_j\}$. Using the
hypothesis that agents cannot influence one another, we can compute the
probability that agents end up playing the profile of strategies $s'\in
R_i(s_i)$. Let $r_i(s')=\{j\in P\setminus\{i\}: s_j'=o_j\}$. The profile
of strategies $s'$ is reached with probability  
\begin{align}
\beta_i(\omega,s'):=\prod_{j\in r_i(s')}\alpha_i(j,\omega)\prod_{j\in P\setminus r_i(s')\setminus\{i\}}(1-\alpha_i(j,\omega)).
\end{align}
}
\item $\PR_i(s,v)(s',v') =\left\{
\begin{array}{lll}
0 &\mbox{if $s_i \ne s'_i$ or $v_i \ne v_i'$,}\\
\PR'_i(s_i)(s'_{-i})\prod_{j\ne i} (v'_j\alpha_i +
(1-v'_j)(1-\alpha_i))
&\mbox{if $s_i = s'_i$, $v_i = v'_i$.}
\end{array}
\right.$
\end{itemize}

We must check that this $G^\G(\sigma,\alpha)$ is a indeed a counterfactual
structure.  Clearly $f$ satisfies CS1 and CS2.  
It is also easy to see from the definition that $\PR_i(s,v)(s',v') > 0$
only if $s_i = s_i'$; moreover, if $s_i = s_i'$, then $\PR_i(s,v) = 
\Pr(s',v')$.  PR1 and PR2 immediately follow once we show that
$\PR_i(s,v)$ is a probability measure, which is done in the following
lemma.  

\begin{lemma}\label{lem}
For all $i\in P$ and all states $(s,v) \in\Omega$, 
$\PR_i(s,v)$ is a probability measure.
\end{lemma}

\begin{proof}
It clearly suffices to show that $\sum_{(s',v') \in \Omega} \PR_i(s,v)
(s',v') = 1$.  Since $\PR_i(s,v)(s',v') = 0$ if $s_i' \ne s_i$, it
suffices to show that $$\sum_{\{(s',v') \in \Omega: s_i' = s_i\}}
\PR_i(s,v) (s',v') = 1.$$  Now 
$$\begin{array}{lll}
&\sum_{\{(s',v') \in \Omega: s_i' = s_i, \, v_i' = v_i\}} \PR_i(s,v) (s',v')\\
= &\sum_{s'_{-i} \in S_{-i}} \sum_{v'_{-i} \in \{0,1\}^{N-1}}
\PR_i(s,v)((s_i,s'_{-i}),(v_i,v_i'))\\
= &\sum_{s'_{-i} \in S_{-i}} \sum_{\{v_{-i}' \in \{0,1\}^{N-1}: \}}
\sigma_{-i}(s'_{-i})
\prod_{\{j\ne i: \, v'_j = 1\}}\alpha_j \prod_{\{j\ne i: \, v'_j =0\{ }(1-\alpha_j )\\
= &\sum_{s'_{-i} \in S_{-i}}  \sigma_{-i}(s'_{-i}) \sum_{\{v_{-i}' \in \{0,1\}^{N-1}: \}}
\prod_{\{j\ne i: \, v'_j = 1\}}\alpha_j \prod_{\{j\ne i: \, v'_j = 0\} }(1-\alpha_j ).
\end{array}
$$

Since $\sigma_{-i}$ is a probability measure on $S_{-i}$, it 
suffices to prove that 
\begin{equation}\label{eq1}
\sum_{v'_{-i} \in \{0,1\}^{N-1}}
\prod_{\{j\ne i: \, v'_j = 1\}}\alpha_j \prod_{\{j\ne i: \, v'_j = 0\}
}(1-\alpha_j ) = 1.
\end{equation}
We proceed by induction on $n$ for $n \ge 2$.  We leave the details to
the reader.  [[I MAY ADD MORE HERE LATER.]]
\commentout{
is $1-\alpha_i$ if $v'_j = 0$.  So we are essentially adding together
all $2^n$ terms of the form $\beta_1 \ldots \beta_N$, where $\beta_j$
is either $\alpha_i$ or $1-\alpha_i$.  
Clearly, if $n=2$, 
$$\sum_{v_{-i} \in \{0,1\}}
\prod_{j\ne i} (v'_j\alpha_i + (1-v'_j)(1-\alpha_i)) = \alpha_i +
(1-\alpha_i) = 1.$$

Suppose that (\ref{eq1}) holds for $n' \ge 2$; we prove it for $n' +
1$.  We assume without loss of generality that $i \ne n'+1$.  Then 
$$\begin{array}{lll}
&\sum_{v'_{-i} \in \{0,1\}^{n'}}
\prod_{j\ne i} (v'_j\alpha_i + (1-v'_j)(1-\alpha_i))\\ 
= &\sum_{v'_{-i} \in \{0,1\}^{n'-1}, v_{n'+1} = 0}
\prod_{j\ne i} (v'_j\alpha_i + (1-v'_j)(1-\alpha_i)) \\
&\ + \sum_{v'_{-i} \in \{0,1\}^{n'-1}, v_{n'+1} = 1}
\prod_{j\ne i} (v'_j\alpha_i + (1-v'_j)(1-\alpha_i)) \\
= &(1-\alpha_i) \sum_{v'_{-i} \in \{0,1\}^{n'}}
\prod_{j\ne i} (v'_j\alpha_i + (1-v'_j)(1-\alpha_i)) \\
& \ + \alpha_i\sum_{v'_{-i} \in \{0,1\}^{n'}}
\prod_{j\ne i} (v'_j\alpha_i + (1-v'_j)(1-\alpha_i)) \\
= &\alpha_i + (1-\alpha_i) \\
= &1.

\end{array}
$$
This completes the proof of (\ref{eq1}) and the lemma.
}
\end{proof}

\section{Applications of the model}\label{se:social dilemmas examples} 
}

We now show that the model described in the previous section allows to explain all patterns observed in experiments on social dilemmas. 

We start by reminding the definition of the major social dilemmas belonging to the class of games that we are considering.

For all these games, Nash equilibrium theory predicts that people should defect. Yet experimental studies have systematically rejected this prediction. Nevertheless, cooperation is not casual and typically depends on the payoffs and on the number of agents. Specifically, the following regularities have been observed.

\begin{itemize}
\item The rate of cooperation in the Prisoner's dilemma depends positively on the benefit $b$ and negatively on the cost $c$ \cite{Rapoport,CJR,EZ}.
\item The rate of cooperation in the Traveler's Dilemma depends negatively on the bonus/penalty $b$ \cite{CGGH}.
\item The rate of cooperation in the Public Goods game depends positively on the constant marginal return \cite{IWT,Gu}.
\item The rate of cooperation in the Public Goods game depends positively on the number of players \cite{IWW,Ze}.
\item The rate of cooperation in the Bertrand Competition depends
  negatively on the number of players \cite{DG02}. 
\end{itemize}
The purpose of this section is to show that the model described in the
previous section predicts all these regularities.

\commentout{
To start, we can say something about the `shape' of the probability $\alpha_i(j,\omega)$ under the assumption that players are anonymous. In state $\omega$, Player $i$ knows only the strategy he or she plays, so \emph{by definition} $\alpha_i(j,\omega)$ cannot depend on the strategies $\textbf{s}_j(\omega)$, with $j\neq i$. Similarly, in state $\omega$ Player $i$ knows that the other players do not know his or her strategy $\textbf{s}_i(\omega)$. Consequently, $\alpha_i(j,\omega)$ does not depend on the strategy $\textbf{s}_i(\omega)$ either. So, essentially, $\alpha_i(j,\omega)$ depends only on the pair of players $(i,j)$. Of course, it would be of great interest to study general situation where players know each other and so this probability will actually depend on the pair $(i,j)$. Under the hypothesis that players are completely anonymous, $\alpha_i(j,\omega)=\alpha_i$ is independent of Player $j$. In this situation, making explicit computations become easier and we indeed obtain that the model proposed predicts all experimental patterns mentioned above. 

\begin{theorem}\label{th:pd}
If $\alpha_1,\alpha_2\leq1-\frac{c}{b}$, then the unique equilibrium of the Prisoner's dilemma is $$\left(\alpha_1D+(1-\alpha_1)C,\alpha_2D+(1-\alpha_2)C\right).$$
\end{theorem}

\begin{proof}
We start by observing that Item 2 in the second condition of Definition \ref{def:equilibrium} implies that the only equilibrium in this model can be $\left(\alpha_1D+(1-\alpha_1)C,\alpha_2D+(1-\alpha_2)C\right)$. Other equilibria are not possible. To show that this is indeed an equilibrium, consider the set $\Omega'\subseteq\Omega$, defined by $$\Omega'=\{(s,(o, o)) : s\in S\}.$$ Properties (1) and (2) in Definition \ref{def:equilibrium} are certainly satisfied. It remains to show that players are rational on each of these states. 

We start by showing that $\omega=((C,D),(o,o))$ is rational for Player 1. Since $\PR_1(\omega)= (1-\alpha_1)((C, C), (o,o)) + \alpha_1((C, D), (o,o))$, one has
\begin{align*}
g_1(\omega)\\
&=\sum_{\omega'\in\Omega}\PR_1(\omega)(\omega')u_1(\mathbf{s}_1(\omega), \mathbf{s}_2(\omega'))\\
&= (1-\alpha_1) u_1(C, C) + \alpha_1u_1(C, D)\\
&= (1-\alpha_1)(b-c)-\alpha_1c\\
&= b(1-\alpha_1) - c
\end{align*}
In order to compute
$$
g_1^c(\omega,D)= \sum_{\omega'\in\Omega}\sum_{\omega'':f(\omega'',1,D)=\omega'} \PR_1(\omega)(\omega'')u_1(D,\mathbf{s}_2(\omega'))
$$
observe that there are only two possible $\omega''$ in the support of $\PR_1(\omega)$, namely $((C, D),(o, o))$ and $((C,C),(o,o))$. Either way, we have $f(\omega'',1,D) = ((D,D),(o,o))$ and so we have only one possibility for $\omega'$. Consequently, $g_1^c(\omega,D)$ is a convex combination of $u_1(D,D) = 0$; thus $g_1^c(\omega,D) = u_1(D,D) = 0$. So, in order for the state $\omega$ to be rational, we need the condition $b(1-\alpha-1)-c\geq0$, which corresponds to the hypothesis of the theorem. A similar computation shows that rationality of the other elements of $\Omega'$ does not add any other condition on $\alpha_1$. Clearly, imposing rationality for player 2 gives symmetric conditions on $\alpha_2$.
\end{proof}

\begin{remark}\label{rem:pd}
{\rm Observe that Theorem \ref{th:pd} agrees qualitatively with the first of the experimental regularities mentioned above. As $b$ increases, $1-\frac{c}{b}$ increase to $1$ and so $\alpha_1,\alpha_2$ are more likely to be smaller than it. So cooperation is easier, as $b$ increases. Analogously, cooperation is more difficult as $c\to b$. }
\end{remark}

\begin{theorem}
If $\alpha_1,\alpha_2\leq1-\frac{b}{H-L+b}$, then the unique equilibrium of the Traveler's Dilemma is $$\left(\alpha_1L+(1-\alpha_1)H,\alpha_2L+(1-\alpha_2)H\right).$$
\end{theorem}

\begin{proof}\label{th:td}
The proof is very similar to the one of the previous theorem. We define $$\Omega'=\{(s,(o,o)) : s_i\in\{L,H\}, \text{ for all } i=1,2\}.$$
Imposing that $\omega=((H,L),(o,o))$ be rational for Player 1 leads to
the condition $\alpha_1\leq1-\frac{b}{H-L+b}$. Analogously, imposing
that $\omega=((L,H),(o,o))$ be rational for Player 2 leads to the
condition $\alpha_2\leq1-\frac{b}{H-L+b}$. As in the previous theorem,
rationality of the other states in $\Omega'$ does not add any more
condition. 
\end{proof}

\begin{remark}\label{rem:td}
{\rm Theorem \ref{th:td} is also consistent with the experimental finding that cooperation is negatively related to the bonus/penalty $b$. As $b$ increases, $1-\frac{b}{H-L+b}$ decreases to 0 and cooperation gets then more and more difficult.}
\end{remark}

\begin{theorem}\label{th:pgg}
If $\alpha_i\leq1-\frac{1}{\rho(N-1)}$, for all $i\in P$, then the unique equilibrium of the N-person Public Goods game with marginal return $\rho$ is the profile of strategies $(\sigma_1,\ldots,\sigma_N)$, where $$\sigma_i=\alpha_iN+(1-\alpha_i)C,$$ and C stands for `contribute everything' and N stands for `contribute nothing'.
\end{theorem}

\begin{proof}
As in the previous proofs, let $$\Omega'=\{(s,(o,o)) : s_i\in\{C,N\}, \text{ for all }i=1,\ldots,n\}.$$ Fix $\omega=((C,N,\ldots,N),(o,\ldots,o))$. We show that Player 1 is rational at $\omega$. Denote $s_1=C$ and recall that $R_1(s_1)$ is the set of strategy profiles where Player $1$ plays $s_1=C$ and all other players either play C or play N. Define the following equivalence relation $\sim$ on $R_1(s_1)$: two profiles of strategies $s'=(s_1,s_{-1}')$ and $s''=(s_1,s_{-1}'')$ are equivalent if the number of players who do not contribute in $s'$ is equal to the number of players who do not contribute in $s''$. Let $k\in\{0,\ldots,N-1\}$ be such a number, which uniquely identify an equivalence class. Let $s'(k)$ be a representative of the class associated with $k\in\{0,\ldots,N-1\}$. 

We start observing that, if $\textbf{s}(\omega')\sim\textbf{s}(\omega'')\sim s'(k)$, then we have
\begin{align}\label{eq:equality}
\beta_1(\omega,\textbf{s}(\omega'))=\beta_1(\omega,\textbf{s}(\omega''))=\alpha_1^k(1-\alpha_1)^{N-1-k}\end{align}
\begin{align}
u_1(\textbf{s}_1(\omega),\textbf{s}_{-1}(\omega'))=u_1(\textbf{s}_1(\omega),\textbf{s}_{-1}(\omega''))=\rho(N-1-k).
\end{align}

Consequently, the partition $\sim$ allows us to write the sum defining $g_1(\omega)$ in a much easier way. We have
\begin{align*}
g_1(\omega)\\
&=\sum_{k=0}^{N-1}\binom{N-1}{k}\beta_1(\omega,s'(k))u_1(s'(k),(o,\ldots,o))\\
&=\sum_{k=0}^{N-1}\binom{N-1}{k}\alpha_1^k(1-\alpha_1)^{N-1-k}\rho(N-1-k)\\
&=\rho(1-\alpha_1)(N-1).
\end{align*}

On the other hand, similarly to the other proofs, we have
$g_1^c(\omega,N)=u_1(N,\ldots,N)=1$. So, rationality of $\omega$ gives
the condition $\alpha_1\leq1-\frac{1}{\rho(N-1)}$. Analogously to the
other proofs, rationality of the other states in $\Omega'$ does not add
any more condition on $\alpha_1$. The analogue conditions on the other
$\alpha$'s are obtained by reasoning with respect to the respective
players.  
\end{proof}

\begin{remark}\label{rem:pgg}
{\rm Also the predictions of Theorem \ref{th:pgg} are consistent with the experimental findings described above. The function $1-\frac{1}{\rho(N-1)}$ increases as either $\rho$ or $n$ increases, favoring cooperation. }
\end{remark}

\begin{theorem}\label{th:bc}
If $\alpha_i\leq1-\left(\frac{L}{H}\right)^{1/(N-1)}$, for all
$i=1,\ldots,n$, then the unique equilibrium of the Bertrand
Competition is the profile of strategies $(\sigma_1,\ldots,\sigma_N)$, where $$\sigma_i=\alpha_iL+(1-\alpha_i)H.$$
\end{theorem}

\begin{proof}
The proof of this results is similar to the proof of Theorem \ref{th:pgg}. The only difference is that now
$$
u_1(\omega,s_{-1}(s'(k))=\left\{
  \begin{array}{lll}
    0, & \hbox{if $k\geq1$} \\
    \frac{H}{n}, & \hbox{if $k=0$},\\
  \end{array}
\right. 
$$
Consequently, one has
\begin{align*}
g_1(\omega)=(1-\alpha_1)^{N-1}\frac{H}{n}.
\end{align*}
On the other hand, one has
$g_1^c(\omega,L)=u_1(L,\ldots,L)=\frac{L}{n}$. So, rationality of
$(H,L,\ldots,L)$ leads to the condition in stated on the theorem. 
\end{proof}

\begin{remark}\label{rem:bc}
{\rm Notice that since $L<H$, then
  $\left(\frac{L}{H}\right)^{1/(N-1)}$ is increasing in $n$ and
  converges to $1$ as $n$ goes off to infinity. Consequently, Theorem
  \ref{th:bc} predicts that cooperation becomes more and more
  difficult as the number of agents increases, consistently with the
  experimental findings. Additionally, observe that Theorem
  \ref{th:bc} also predicts that cooperation is more difficult as the
  price floor $L$ increases. We are aware of only one experimental
  study on the Bertrand Competition with varying price floor and,
  indeed, the authors found that cooperation is more difficult as $L$
  increases \cite{D07}.} 
\end{remark}
}
}

As we suggested in the introduction, we hope to use translucency to
explain cooperation in social dilemmas.  To do this, we have to make
assumptions about an agent's beliefs.  
Say that an agent $i$ has
\emph{type $(\alpha,\beta,C)$} if $i$ intends to cooperate (the
parameter $C$ stands for \emph{cooperate}) and 
believes that (a) if he deviates from that, then each other 
agent will independently realize this with probability $\alpha$; 
(b) if an agent $j$ realizes that $i$ is not going to cooperate,
then $j$ will defect;
and (c) all other players will either cooperate or defect, and they
will cooperate with probability $\beta$.

The standard assumption, of course, is that $\alpha = 0$.  Our results
are only of interest if $\alpha > 0$. The assumption that $i$ believes
that agent $j$ will defect if she realizes that $i$ is going to
deviate from cooperation seems reasonable; defection is the ``safe''
strategy.  We stress that, for our results, it does not matter what
$j$ actually does.  All that matters are $i$'s beliefs about what $j$
will do.  The assumption that players will either cooperate or defect
is trivially true in Prisoner's Dilemma, but is a highly nontrivial
assumption in the other games we consider.  While cooperation and
defection are arguably the most salient strategies, we do in practice
see players using other strategies. 
For instance, the distribution of strategies in the Public Goods game
is typically tri-modal, concentrated on contributing nothing,
contributing everything, and contributing half \cite{CJR}. 
We made this assumption mainly for technical convenience; it
makes the calculations much easier.  We 
believe that results qualitatively similar to ours will hold under a
much weaker assumption, namely, that a type $(\alpha,\beta,C)$ player
believes that other players will cooperate with probability $\beta$
(without assuming that they will defect with probability $1-\beta$).

Similarly, the assumptions that a social dilemma has a unique Nash
equilibrium and a unique social-welfare maximizing strategy were made
largely for technical reasons.  We can drop these assumptions,
although that would require more complicated assumptions about
players' beliefs.

The key feature of our current assumptions is that 
the type of player $i$ determines the distributions $\mu_i^{s_i,s_i'}$.
In a social dilemma with $N$ agents,
the distribution $\mu_i^{s_i,s_i}$ assigns probability
$\beta^r(1-\beta)^{N-1-r}$ to a strategy profile $s_{-i}$ for the
players other than $i$ if exactly $r$ players cooperate in $s_{-i}$
and the remaining $N-1-r$ players defect; it assigns probability 0 to all
other strategy profiles.  The distributions
$\mu_i^{s_i,s_i'}$ for $s_i' \ne s_i$ all have the form 
$\sum_{J \subseteq \{1,\ldots,i-1, i+1,\ldots, N\}} \alpha^{|J|}
(1-\alpha)^{N-1-|J|} \mu_i^{J}$, where $\mu_i^J$ is the distribution
  that assigns probability $\beta^k(1-\beta)^{N-|J|-k}$ to a profile
  where $k \le N-1 - 
  |J|$ players not in $J$ cooperate, and the remaining players (which
  includes all the players in $J$) defect.  Thus, $\mu_i^J$ is the
  distribution that describes what player $i$'s beliefs would
  be if he knew that exactly the players in $J$ had noticed his
  deviation (which happens with probability $\alpha^{|J|}
(1-\alpha)^{N-1-|J|}$). 
In the remainder of this section, when we talk about best response, it
is with respect to these beliefs.


For our purposes, it does not matter where the beliefs $\alpha$ and $\beta$
that make up a player's type come from.  We do not assume, for
example, that other players are 
(translucently)
rational.  For example, $i$ may believe that some players cooperate
because they are altruistic, while others may cooperate because they have
mistaken beliefs.  We can
think of $\beta$ as summarizing $i$'s previous experience of
cooperation when playing social dilemmas.  Here we are interested in
the impact of the parameters of the game on the reasonableness of
cooperation, given a player's type.

The following four propositions analyze the four social dilemmas in
turn.  We start with Prisoners Dilemma.  Recall that $b$ is the
benefit of cooperation and $c$ is its cost.

\begin{proposition}\label{prop:PD}
In Prisoner's Dilemma,  it is translucently rational for a player of 
type $(\alpha,\beta,C)$ to cooperate
if and only if 
$\alpha \beta b \ge c$.
\end{proposition}

\begin{proof}
If player $i$ has type $(\alpha,\beta,C)$ and cooperates in Prisoner's
Dilemma, then his expected
payoff is $\beta(b-c) - (1-\beta)c$, since player $i$ believes that
$j \ne i$
will cooperate with 
probability $\beta$.  However, if $i$ deviates from his intended
strategy of cooperation, then $j$ will catch him with probability
$\alpha$ and also defect.  Thus, if $i$ deviates, then $i$'s belief
that $j$ will cooperate goes down from $\beta$ to
$(1-\alpha)\beta$.  
(We remark that this is the case in all social dilemmas; this fact
will be used in all our arguments.)
This means that $i$'s expected payoff if he
deviates by defecting is
$(1-\alpha)\beta b$.  So cooperating
is a best response if $\beta(b-c) - (1-\beta)c \ge (1-\alpha)\beta b$.   A
little algebra 
shows that this reduces to $\alpha \beta b \ge c$. 
\end{proof}

As we would expect, if $\alpha = 0$, then cooperation is not
a best response in Prisoner's Dilemma; this is just the standard
argument that defection 
dominates cooperation.   But if $\alpha > 0$, then cooperation can be rational.
Moreover, if we fix $\alpha$,  
the greater the benefit of cooperation and the smaller the cost, then
the smaller the value of $\beta$ that still allows cooperation to be a
best response.

\commentout{
\begin{proposition} If $\Gamma$ is prisoner's dilemma, and 
$\sigma = (p_1C + (1-p_1)D, p_2C + (1-p_2)D)$.  Then
$\sigma$ is a translucent equilibrium in $M^{\G}(\sigma,\alpha)$ if 
$\min(\alpha_1 b p_1, \alpha_2, b p_2) \ge c$.
\end{proposition}

\begin{proof} Take $\Omega'$ to consist of all states $(s,v)$ such
that $s$ is in the support of $\sigma$.  IT is immediate that TE1,
TE2, and TE3 hold, so we just have to check TE4.  Clearly if a player
defects at state $\omega$, this is rational.  So we just have to show
that cooperation is rational for both player 1 and player 2.
\end{proof}
}

We next consider Traveler's Dilemma.  Recall that $b$ is the
reward/punishment, and $H$ and $L$ are the high and low payoffs,
respectively. 



\begin{proposition}\label{prop:TD}
In Traveler's Dilemma, 
it is translucently rational for a player of $(\alpha,\beta,C)$ to cooperate
if and only if
$b\le
\left\{
\begin{array}{ll}
\frac{(H-L)\beta}{1-\alpha\beta} &\mbox{if $\alpha\ge\frac12$}\\
\min\left(\frac{(H-L)\beta}{1-\alpha\beta},\frac{H-L-1}{1-2\alpha}\right) &\mbox{if $\alpha<\frac12$;}
\end{array}
\right.$

\end{proposition}

\begin{proof}
If player $i$ has type $(\alpha,\beta,C)$ and cooperates in Traveler's
Dilemma, then his expected
payoff is $\beta H + (1-\beta)(L-b)$, since player $i$ believes that
$j \ne i$
will cooperate with 
probability $\beta$.  If $i$ deviates and plays $x\ne H$, then $j$ will catch him with probability
$\alpha$ and play $L$.  Recall from the proof of
Proposition~\ref{prop:PD} that, if $i$ deviates, $i$'s belief that $j$
cooperates is $(1-\alpha)\beta$. 
This means that $i$'s expected payoff if he
deviates to $x < H$ is
$(1-\alpha)\beta (x+b) + (1 - \beta + \alpha\beta)(L-b)$ if $x > L$, and 
$(1-\alpha)\beta (L+b) + (1- \beta + \alpha \beta)L = L +
(\beta -\alpha\beta)  b$ if $x = L$.
It is easy to see that $i$ maximizes his expected payoff either if $x =
H-1$ or $x=L$.  Thus, cooperation is a best response if 
$\beta H + (1-\beta)(L-b) \ge \max((1-\alpha)\beta(H +b -1) + (1-
\beta + \alpha \beta)(L-b), L + (\beta -\alpha\beta)  b)$. 
Again, straightforward algebra shows that this condition is equivalent
to the one stated, as desired.
(It is easy to check that if $\alpha \ge 1/2$, then the condition 
$\beta H + (1-\beta)(L-b) \ge (1-\alpha)\beta(H +b -1) + (1- \beta +
\alpha \beta)(L-b) $ is guaranteed to hold, which is why we get the
two cases depending on whether $\alpha \ge 1/2$.)
\end{proof}

Proposition~\ref{prop:TD} shows that as $b$, the punishment/reward,
increases, a player must have greater belief that his opponent is
cooperative and/or a greater belief that the opponent will learn about
his deviation and/or a
greater difference between the high and low payoffs in order to make
cooperation a best response.  (The fact that 
increasing $\beta$ increases $\frac{(H-L)\beta}{1-\alpha\beta}$
follows from straightforward calculus.)

\commentout{
recall the identity

\begin{align}\label{eq:useful identity}
\sum_{k=0}^{n}\binom{n}{k}p^k(1-p)^{N-k}k=pn
\end{align} 
}

We next consider the Public Goods game.  Recall the $\rho$ is 
the marginal return of cooperating.

\begin{proposition}\label{pro:PCG}
In the Public Goods game with $N$ players, it is translucently rational
for a player of type 
$(\alpha,\beta,C)$ to cooperate
if and only if 
$\alpha\beta\rho(N-1) \ge 1 - \rho$.
\end{proposition}

\begin{proof}
Suppose player $i$, of type $(\alpha,\beta,C)$, cooperates.  Since he expects a player to cooperate
with probability $\beta$, the expected number of cooperators among the
other players is $\beta(N-1)$. Since he himself will cooperate, the
total expected number of cooperators is $1+\beta(N-1)$.  
Since $i$'s payoff is $\rho m$ if $m$ players
including him cooperate, and thus is linear in the number of
cooperators, his  expected payoff is exactly his payoff if the
expected number of players cooperate.  
Since his expected payoff with $1 + \beta(N-1)$ cooperators is 
$\rho(1 + \beta(N-1))$, this is his expected payoff if he cooperates.   

On the other hand, if $i$ deviates by contributing $x < 1$, his
expected payoff if $m$ other players cooperate is $(1-x) + \rho
(m+x)$.  
Again, if $i$ deviates, his expected belief that $j$ will cooperate is 
$(1-\alpha)\beta$.
Thus, the expected
number of cooperators is $(1-\alpha)\beta(N-1)$, and
his expected payoff is $1 - x + \rho((1-\alpha)\beta(N-1) + x)$.  
Since $\rho < 1$, he gets the highest expected payoff by defecting
(i.e., taking $x=0$).  

Thus, cooperation is a best response if $\rho(1 + \beta(N-1)) \ge 1 +
\rho(1-\alpha)\beta(N-1)$.    Simple algebra shows that this condition
holds iff $\alpha\beta\rho(N-1) \ge 1-\rho$.
\end{proof}

Proposition~\ref{pro:PCG} shows that if $\rho=1$, then cooperation is certainly
a best response (you always get out at least as much as you contribute).  For
fixed $\alpha$ and $\beta$, there is guaranteed to be an $N_0$ such
that cooperation is a best response for all $N \ge N_0$; moreover, for
fixed $\alpha$, as $N$ gets larger, smaller and smaller $\beta$s are
needed for cooperation to be a best response.

Finally we consider the Bertrand competition. Recall that $H$ is the reservation value and $L$ is the price floor.

\begin{proposition}\label{pro:Bertrand}
In Bertrand Competition, 
it is translucently rational for a player of type $(\alpha,\beta,C)$
to cooperate iff $\beta^{N-1} \ge f(\gamma,N)LN/H$, where 
$f(\gamma,N) = \sum_{k=0}^{N-1}
\binom{N-1}{k}(1-\gamma)^k\gamma^{N-k-1}/(k+1)$ and  
$\gamma = (1-\alpha)\beta$.
\end{proposition}

\begin{proof}
\commentout{
One has 
$$
EU_1((H,L,\ldots,L),(1,0,\ldots,0))=\frac{p_{-1}^{N-1}H}{n}
$$
Similar to the previous proposition, we then find

$$
EU_1(\omega,L)=

$$

$$
\sum_{r=0}^{N-1}\sum_{t=0}^{r}\sum_{w=0}^{N-1-r}\binom{N-1}{r}\binom{r}{t}\binom{N-1-r}{w}p_{-1}^{t+w}(1-p_{-1})^{N-1-t-w}\alpha_1^{N-1-r}(1-\alpha_1)^{r}\frac{L}{N-t}
$$

This time, this expression is not easily simplificable (using wolfram). But we can complete the proof in the following way: fixing $t=N-1$ and $w=0$ we obtain

$$
EU_1(\omega,D)\geq p_{-1}^{N-1}L\sum_{r=0}^{N-1}\binom{N-1}{r}\alpha_1^{N-1-r}(1-\alpha_1)^r=p_{-1}^{N-1}L
$$

Consequently, for $N$ large enough , we certainly have
$EU_1(\omega,D)>EU_1(\omega)$. 
}
Clearly, if player $i$ cooperates, then his expected payoff is
$\beta^{N-1}H/N$, since he gets $H/N$ if everyone else cooperates
(which happens with probability $\beta^{N-1}$), and otherwise gets
$0$.

Let $\gamma = (1-\alpha)\beta$.  Again, this is the probability that
$i$ ascribes to another player playing $H$ if he deviates.
If $i$ deviates, then it is easy to see (given his beliefs) that the
optimal choices for deviation are $H-1$ and $L$.  In the former case, 

$i$'s expected payoff is $\gamma^{N-1}(H-1)$.  
In the latter case, 
$i$'s expected payoff is $\sum_{k=0}^{N-1}\binom{N-1}{k} (1-\gamma)^k\gamma^{N-k-1} L/(k+1)$:
with probability  $(1-\gamma)^k\gamma^{N-k-1}$, exactly $k$ other
players will play $L$, and $i$'s payoff will be $L/(k+1)$.  Moreover, each possible subset of $k$ defectors, has to be count $\binom{N-1}{k}$ times.
Let $f(\gamma,N) = \sum_{k=0}^{N-1}\binom{N-1}{k} (1-\gamma)^k\gamma^{N-k-1}/(k+1)$.
Note that, as the notation suggests, this expression depends only on
$\gamma$ and $N$ (and not any of the other parameters of the game).
Thus, $i$'s expected payoff in this case is $f(\gamma,N)L$, so 
cooperation is a best response iff 
$\beta^{N-1}H/N \ge \max(\gamma^{N-1}(H-1), f(\gamma,N)L)$.
 While it seems
difficult to find a closed-form expression for $f(\gamma,N)$,
this does not matter for our purposes.%
\footnote{Note that the expected value of $L/(k+1)$ cannot be computed by
plugging in the expected value of $k$, in the spirit of our earlier
calculations, since $L/(k+1)$ is not linear in $k$.} 
%
Since we clearly have $\beta^{N-1}H/N \ge \gamma^{N-1}(H-1)$,
cooperation is a best response iff $\beta^{N-1}H/N \ge f(\gamma,N)L$,
or, equivalently, $\beta^{N-1} \ge f(\gamma,N)LN/H$,
\commentout{
In the latter case, 
$i$'s expected payoff is $\frac{L}{(1-\gamma)(N-1)+1}$, since the expected number of players other than $i$ playing $L$ is $(1-\gamma)(N-1)$.
Thus, cooperation is a best response iff 
$\beta^{N-1}H/N \ge \max(\gamma^{N-1}(H-1), \frac{L}{(1-\gamma)(N-1)+1})$.
Since we clearly have $\beta^{N-1}H/N \ge \gamma^{N-1}(H-1)$,
cooperation is a best response iff $\beta^{N-1}H/N \ge \frac{L}{(1-\gamma)(N-1)+1}$,
or, equivalently, $\beta^{N-1} \ge \frac{LN}{H((1-\gamma)(N-1)+1)}$,} 
\end{proof}

Note that $f(\gamma,N) = \sum_{k=0}^{N-1}\binom{N-1}{k} (1-\gamma)^k\gamma^{N-k-1}/(k+1) \ge 
\sum_{k=0}^{N-1}\binom{N-1}{k} (1-\gamma)^k\gamma^{N-k}/N = 1/N$, so 
Proposition~\ref{pro:Bertrand} shows cooperation is irrational if
$\beta^{N-1} < L/H$.
Thus, while cooperation may be achieved for reasonable values of $\alpha$ and 
$\beta$ if $N$ is small, a
player must be more and more certain of cooperation in order to
cooperate in Bertrand Competition as the number of players
increases.  Indeed, for a fixed type $(\alpha,\beta,C)$, there exists
$N_0$ such that cooperation is not a best response for all $N \ge N_0$.
Moreover, if we fix the number $N$ of players, more values of $\alpha$
and $\beta$ allow cooperation as $L/H$ gets smaller.  In particular,
if we fix $H$ and raise the floor $L$, fewer values of $\alpha$ and
$\beta$ allow cooperation.

While Propositions~\ref{prop:PD}--\ref{pro:Bertrand} are suggestive, 
we need to make extra assumptions to use these propositions to make
predictions.  A simple assumption that suffices is that there are a
substantial number of translucently rational players whose types have the form
$(\alpha,\beta,C)$, and for each pair  $(u,v)$ and $(u',v')$ of open
intervals in $[0,1]$, there is a positive probability of finding
someone of type $(\alpha,\beta,C)$ with $\alpha \in (u,v)$ and $\beta
\in (u',v')$.  With this assumption, it is easy to see that all
the regularities discussed in Section~\ref{se:social dilemmas model}
hold.

\section{Discussion}\label{sec:discussion}

We have presented an approach that explains a number of well-known
observations regarding the extent of cooperation in social dilemmas.  
In addition, our approach can also be applied to explain the apparent
contradiction that people cooperate more in a one-shot Prisoner's
dilemma when they do not know the other player's choice than when they do.
In the latter case, Shafir and Tversky
\citeyear{ShafirTversky92} found that most people (90\%) defect, while
in the former case, only 63\% of people defect. Our model of
translucent players predicts this behavior: if player 1 knows player 2
choices, then there is no translucency and thus our model predicts
that player 1 defects for sure. On the other hand, if player 1 does
not know player 2's choice and believes that he is to some extent
translucent, then, as shown in Proposition \ref{prop:PD}, he may be
willing to cooperate. Seen in this light, our model can also be
interpreted as an attempt to formalize \emph{quasi-magical thinking}
\cite{ShafirTversky92}, the kind of reasoning that is supposed
to motivate those people who believe that the others' reasoning is
somehow influenced by their own thinking, even though they know that
there is no causal relation between the two. Quasi-magical thinking
has also been formalized by Masel \citeyear{Masel} in the context of
the Public Goods gam and by Daley and Sadowski
\citeyear{DaleySadowski} in the context of symmetric $2\times 2$
games.  The notion of translucency goes beyond these models,
since it may applied to a much larger set of games. 

Besides a retrospective explanation, our model makes new predictions for
social dilemmas which,
to the best of our knowledge, have never been tested in the lab. 
In particular, 
it predicts that 
\begin{itemize}
\item the degree of cooperation in Traveler's dilemma increases as the
  difference $H-L$ increases; 
\item for fixed $L$ and $N$, the degree of cooperation in Bertrand
  Competition increases as $H$ increases, and what really matters is
  the ratio $L/H$.  
\end{itemize}

Clearly much more experimental work needs to be done to validate
the approach.  
For one thing, it is important to understand the predictions it makes 
for other social dilemmas and for games that are not social
dilemmas.  Perhaps even more important would be to see if we can
experimentally verify that people believe that they are to some extent
translucent, and, if so, to get a sense of what the value of $\alpha$
is.  In light of the work on watching eyes mentioned in the
introduction, it would also be interesting to know what could be done
to manipulate the value of $\alpha$.

\commentout{
Do we really need to assume that the game has only one Nash
equilibrium? As we mentioned in Section \ref{sec:explanation}, our
results would have been qualitatively the same if we had not assumed
that players believe that the other players will either cooperate or
defect. As a consequence of this, our results would be qualitatively
the same if we had more Nash equilibria, since we may assume that each
strategy that is part of a Nash equilibrium is played with some
probability. For instance, in the Bertrand Competition with $L=1$,
both $s=1$ and $s=2$ are part of a Nash equilibrium. We may assume
that player $i$ believes that player $j$ cooperates with probability
$\beta$ and plays $s=1$ with probability $\beta_1$ and $s=2$ with
probability $\beta_2$ ($\beta_1+\beta_2=1-\beta$).  
}

As we mentioned, there have been many other attempts to explain
cooperation in social dilemmas, especially recently.  
Most of other approaches that we are aware of are not able to obtain
all the regularities that we have mentioned. 
\begin{itemize}
\item The Fehr and Schmidt \citeyear{Fe-Sc} inequity-aversion model
assumes that subjects play a Nash equilibrium of a modified game, in
which players do not only care about their monetary payoff, but also
they care about equity. Specifically, player $i$'s utility when
strategy $s$ is played is assumed to be 
\newcommand{\alphaEF}{a^{FS}}
\newcommand{\betaEF}{b^{FS}}
$U_i(s)=u_i(s)-\frac{\alphaEF_i}{N-1}\sum_{j\neq
  i}\max(u_j(s)-u_i(s),0)-\frac{\betaEF_i}{N-1}\sum_{j\neq
  i}\max(u_i(s)-u_j(s),0)$, where $u_i(s)$ is the material payoff of
player $i$, and $0\leq\betaEF_i\leq\alphaEF_i$ are individual parameters,
where $\alphaEF_i$
represents the extent to which player $i$ is averse to inequity in
favor of others, and $\betaEF_i$ represents his
aversion to inequity in his favor.
\commentout{
It is easy to check that if the parameters $\alphaEF_i$ and $\beta_i$
are such that $U_i(C,C,\ldots,C)\geq U_i(C,\ldots,C,D,C,\ldots,C)$ for
some set of $N$ players, then this inequality holds for all players.
It follows that, if we fix a pool of $n$
players and we randomly extract two samples of $N_1$ and $N_2$
subjects and have them play the PGG, then mutual cooperation is an
equilibrium in the first group if and only if it is in the second
group.  Consequently, Fehr and Schmidt's model does not predict an 
effect of the group size on cooperation in the Public Goods game. 
Consider the 2-player Bertrand Competition.  The payoff for both
players of $(k,k)$ is $k/2$.  If $k > L$ and player $i$ deviates to
$k-1$, then his 
payoff is $(1-\betaEF_i)(k-1)$.  Thus, $(k,k)$ is an equilibrium in
the Fehr-Schmidt model iff $k=L$ or  $k/2 \ge \max_{i=1,2}(1-\betaEF_i)(k-1)$.  
Easy algebra shows that the latter condition holds iff 
$(1/2 - \betaEF_i)k \le  1-\betaEF_i$ for $i = 1,2$.  Thus, if
$\betaEF_i \ge 1/2$ for $i=1,2$, then $(k,k)$ is an equilibrium for
all $k$.  If $\betaEF_i < 1/2$ for some $i$ then we may get an
equilibrium for small values of $k$, but the equilibrium for $k > L$
will disappear once the floor is raised sufficiently high.)}
Consider the Public Goods game with $N$ players. The strategy profile
$(x,\ldots,x)$, where all players contribute $x$ gives player $i$ a
utility of $(1-x) + \rho Nx$.  If $x > 0$ and player $i$ contributes
$x' < x$, then 
his payoff is $(1-x') + \rho((N-1)x + x') - \betaEF_i \rho (x-x')$.
Thus, $(x,\ldots, x)$ is an equilibrium if $\betaEF_i \rho (x-x') \ge
(1-\rho)(x-x')$, that is, if $\betaEF_i \ge (1-\rho)/\rho$.  Thus,
if $\betaEF_i \ge (1-\rho)/\rho$ for all players $i$, then $(x,\ldots,
x)$ is an equilibrium 
for all choices of $x$ and all values of $N$. While there may be other
pure and mixed strategy equilibria, it is not
hard to show that if $\betaEF_i < (1-\rho)/\rho$, then player $i$ will
play 0 in every equilibrium (i.e., not contribute anything). 
As a consequence, assuming, as in our model, that players believe that
there is a probability $\beta$ that other agents will cooperate and
that the other agents either cooperate or defect, Fehr and Schmidt
\citeyear{Fe-Sc} model does not make any clear prediction of a 
group-size effect on cooperation in the public goods game.

\item  McKelvey and Palfrey's \citeyear{MK-Pa95} \emph{quantal response
  equilibrium (QRE)} is defined as follows.%
\footnote{We actually define here   a particular instance of QRE
  called the \emph{logit QRE}; $\lambda$ is a free parameter of this model.}
Taking $\sigma_i(\ssigma)$ to be the probability that mixed strategy
$\sigma_i$ assigns to the pure strategy $\ssigma$, given $\lambda>0$, 
a mixed strategy profile $\sigma$ is a QRE if, for each player $i$, 
$\sigma_i(\ssigma) = \frac{e^{\lambda 
      EU_i(\ssigma,\sigma_{-i})}}{\sum_{s_i'\in S_i}e^{\lambda
    EU_i(s_i',\sigma_{-i})}}$. 

To see that QRE does not describe human behaviour well in social
dilemmas, observe that in the Prisoner's Dilemma, for all choices of
parameters $b$ and $c$ in the game, all choices of the
parameter $\lambda$, all players $i$, and all (mixed) strategies $s_{-i}$ of player
$-i$, we have $EU_i(C,s_{-i})<EU_i(D,s_{-i})$. Consequently, whatever
the QRE $\sigma$ is, we must have 
$\sigma_i(C)<\frac{1}{2}<\sigma_i(D)$, that is, QRE predicts
that the degree of cooperation can never be larger than 50\%. However,
experiments show that we can increase the benefit-to-cost ratio so as to reach
arbitrarily large degrees of cooperation (close to 80\% in \cite{CJR}
with $b/c=10$). 

\item \emph{Iterated regret minimization} \cite{HP11b} does not make appropriate
predictions in Prisoner's Dilemma and the Public Goods game, because
it predicts that if there is a dominant strategy then it will be
played, and in these two games, playing the Nash equilibrium is the
unique dominant strategy. 
\item Capraro's \citeyear{Ca} notion of \emph{cooperative equilibrium}, 
while correctly predicting the effects of
the size of the group on cooperation in the Bertrand Competition and
the Public Goods game \cite{BarceloCapraro}, fails to predict
the negative effect of the price floor on cooperation in the Bertrand
Competition.  
\item Rong and Halpern's \citeyear{RH} notion of \emph{cooperative
  equilibrium} (which is different from that of Capraro \citeyear{Ca})
focuses on 2-player games.  However, the definition for
  games with greater than 2 players does not predict the decrease in
  cooperation as $N$ increases in Bertrand Competition, nor the
  increase as $N$ increases in the Public Goods Game.
\end{itemize}
\newcommand{\alphaCR}{a^{CR}}
\newcommand{\deltaCR}{b^{CR}}
The one approach besides ours that we are aware of that obtains all the 
regularities discussed above is that of 
Charness and Rabin \citeyear{Ch-Ra}.  Charness and Rabin, like Fehr
and Schmidt \citeyear{Fe-Sc},
assume that agents play a Nash equilibrium of a modified game, where
players care not only about their personal material payoff, but 
also about the social welfare and the outcome of the least
fortunate person. Specifically, player $i$'s utility is assumed to be
$(1-\alphaCR_i)u_i(s)+\alphaCR_i(\deltaCR_i\min_{j=1,\ldots,N}u_j(s)+(1-\deltaCR_i)\sum_{j=1}^Nu_j(s))$.
Assuming, as in our model, that agents believe that other players
either cooperate or defect and that they cooperate with probability
$\beta$, then it is not hard to see that Charness and Rabin
\citeyear{Ch-Ra} also predict all the regularities that we have been
considering.  
\commentout{
Consider the 2-player Bertrand Competition.  If $(k,k)$ is played, then
player $i$ gets a utility of $(1-\alphaCR_i)k/2 + 
 \alphaCR_i (\deltaCR_i k/2 + (1-\deltaCR_i) k)$; if $k > L$ and 
player $i$ deviates to $k-1$, then his utility is 
$(1-\alphaCR_i)(k-1) +  \alphaCR_i (1-\deltaCR_i k) (k-1)$.
Straightforward manipulations show that $(k,k)$ is an equilibrium if 
\begin{equation}\label{eqCR}
k(1-\alphaCR_i-\alphaCR_i\deltaCR_i) \le 1-\alphaCR_i\deltaCR_i
\end{equation}
for $i = 1, 2,$. It follows that if $1
\le\alphaCR_i-\alphaCR_i\deltaCR_i$, then $(k,k)$ is an
equilibrium for all $k$.  However, if $1
> \alphaCR_i-\alphaCR_i\deltaCR_i$, then $(k,k)$ is an equilibrium
only for sufficiently small $k$, those for which (\ref{eqCR}) holds.
Though there may be other equilibria,
this can be viewed as consistent with the observation that we see less
cooperation as $L$ increases.  
}

Although it seems difficult to distinguish our model from that of 
Charness and Rabin \citeyear{Ch-Ra} if we consider only social
dilemmas, they are distinguishable if we look at other settings and take
into account the other reason we mentioned for translucency: that
other people in their social group might discover how they acted.
We can easily capture this in the framework we have been considering
by doubling the number of agents; for each player $i$, we add
another player $i^*$ that represent's $i$'s social network.  Player
$i^*$ can play only two actions: $n$ (for ``did \emph{n}ot observe
player $i$'s action) and $o$ (for ``\emph{o}bserved player $i$'s
action'').%
\footnote{Alternatively, we could take player $i$'s payoff to depend
  on the state of the world, where the state would model whether or
  not player $i$'s action was observed.}  The payoffs of these new
players are irrelevant.  Player $i$'s payoff depends on the action of
player $i^*$, but not on the actions of player $j^*$ for $j^* \ne
i^*$.  Now player $i$ must have a prior probability $\gamma_i$ about
whether his action will 
be observed; in a social dilemma, this probability might increase to
$\gamma_i' \ge \gamma_i$ if he intends to cooperate but instead
deviates and defects.  
It should be clear that, even if $\gamma_i' = \gamma_i$,
if we assume that player $i$'s utilities are significantly lower if
his non-cooperative action is observed, with this framework we would
get qualitatively 
similar results for social dilemmas to the ones that we have already
obtained. 
The advantage of taking into account what your social group thinks is
that it can be applied even to single-player games like the
Dictator Game \cite{KKT86}.  To do so, we would need to think about what
a player's utility would be if his social group knew the extent to which
he shared the pot.  But it should be clear that reasonable assumptions
here would lead to some degree of sharing.  

While this would still not distinguish our predictions from those of the
Charness-Rabin model, there is a variant of the Dictator Game
considered by Capraro \citeyear{Capraro14} that does allow us to
distinguish between the two.  In this game,
there are only two possible allocations of money: either the agent
gets $x$ and the other players gets $-x$, or the 
other player gets $x$ and 
the agent gets $-x$.  In this game, the Charness-Rabin approach would
predict that the agent will choose to keep the $x$.  But the translucency
approach would allow that there would be types of agents who would
think that their social group would approve of them giving away $x$, so,
if the action was observed by their social group, they would get high
utility by giving away $x$.  And, indeed, Capraro's results show that 
a significant fraction (25\%) of people do choose to give away $x$.

Of course, we do not have to assume $\alpha > 0$ to get cooperation in 
social dilemmas such as 
Traveler's Dilemma or Bertrand Competition.  But 
we do if we want to consider what we believe is the appropriate
equilibrium notion.   
Suppose that 
rational
players are chosen at random from a population and play a social
dilemma.  Players will, of course, then update their beliefs about the
likelihood of seeing cooperation, and perhaps change their strategy as
a consequence. Will these beliefs stabilize and the strategies played
stabilize?  By \emph{stability} here, we mean that
(1) players are all best responding to their beliefs, 
and (2) players' beliefs about the strategies played by
others are correct: if player $i$ ascribes probability $p$ to player
$j$ playing a strategy $\ssigma_j$, then in fact a proportion $p$ of
players in the population play $\ssigma_j$.   
We have deliberately been fuzzy here about whether we mean best
response in the sense of Definition~\ref{brstandard} or 
Definition~\ref{brdefinitionnew}.  If we use
Definition~\ref{brstandard} (or, equivalently use 
Definition~\ref{brdefinitionnew} and take $\alpha = 0$), then 
it is easy to see (and
well known) that the only way that this can happen is if the
distribution of strategies played by the players represents a mixed
strategy Nash equilibrium.  On the other hand, if $\alpha > 0$
and we use  Definition~\ref{brdefinitionnew}, then
we can have stable beliefs that accurately reflect the strategies used
and have cooperation (in all the other social
dilemmas that we have studied).  We make this precise in the appendix
using the framework of Halpern and Pass \citeyear{HaPa13}, 
by defining a notion of \emph{translucent equilibrium}.  Roughly
speaking, we construct a model where, at all states, players are 
translucently rational (so we have common belief of translucent
rationality), the strategies 
used are common knowledge, and we nevertheless
have cooperation at some states.  
Propositions~\ref{prop:PD}--\ref{pro:Bertrand} play a key role in
this construction; indeed, as long as the strategies used satisfy the 
constraints imposed by these results, we get a translucent equilibrium.


We have not focused on translucent equilibrium here
in the main text because it makes
strong assumptions about players' rationality and beliefs (e.g., it
implicitly assumes common belief of translucent rationality).  We do
not need such strong assumptions for our results.

\commentout{
Another issue that we have not addressed here is a solution concept
based on translucency.  This becomes particularly relevant if a game
is played repeatedly and most of the players are rational.  For the
rest of this discussion, we assume that all players are rational, just
to simplify this presentation.

 Suppose that, at each round, players are chosen at random from a
 population and play a social dilemma.
Players will, of course, then update their beliefs about the
likelihood of seeing cooperation.  Will these beliefs stabilize after
a while, so that all players are players have the same beliefs and are
playing a best response to their beliefs, under the assumption that
$\alpha = 0$ (so that the notion of best response is that given in 
Definition~\ref{brstandard})?  
}


\appendix

\section{Translucent equilibrium}
In the main text of this paper we have described how cooperation can
be rational if players are translucent, that is, if they believe that
if they switch from one strategy to another, the fact that they choose
to switch may be visible to the other players. In this appendix, we
show how to use counterfactual structures to define a notion of
equilibrium with translucent players and we observe that rationality
of cooperation shown in the main text corresponds to having a mixed
strategy translucent equilibrium, where cooperation is played with
non-zero probability.   
We start by reviewing the relevant definitions from \cite{HaPa13}.

\subsection{Game theory with translucent players}

Let $\G=\G(P,S,u)$ be a (finite) normal form game, where
$P=\{1,\ldots,N\}$ is the set of players, each of which has finite
pure strategy set $S_i$ and utility function $u_i$. 

\begin{definition}\label{defin:counterfactual structure}
{\rm A finite counterfactual structure appropriate for the game
$\mathcal G$ is a tuple $M = (\Omega, \mathbf{s}, f, \PR_1,\ldots,\PR_N)$,
where: 
\begin{itemize}
\item $\Omega$ is a finite space of states;
\item $\mathbf{s}:\Omega\to \SSigma$ is the function that associates to each state $\omega$ the strategy profile that is supposed to be played at $\omega$;
\item $f$ is the closest-state function, which describes what would
happen if player $i$ switched strategy to $s_i'$ at state $\omega$. Thus, $f:\Omega\times P\times S_i\to\Omega$ has to verify the following properties:
\begin{itemize}
\item[CS1.] $\mathbf{s}_i(f(\omega,i,\ssigma'))=\ssigma_i'$;
\item[CS2.] $f(\omega,i,\mathbf{s}_i(\omega))=\omega$.  
\end{itemize}
Property CS1 assures that, at state
$f(\omega,i,\ssigma_i')$, player $i$ plays $\ssigma_i'$, and Property CS2 assures that the state does not change if player $i$ does not change strategy.
\item $\PR_i$ are player $i$'s beliefs, which depends on the state $i$ is reasoning about. Specifically, for each $\omega\in\Omega$, $\mathcal P\mathcal R_i(\omega)$ is a
probability measure on $\Omega$ satisfying the following properties:
\begin{itemize}
\item[PR1.] $\PR_i(\omega)(\{\omega'\in\Omega :
\mathbf{s}_i(\omega')=\mathbf{s}_i(\omega)\})=1$ (where
$\strat_i(\omega)$ denotes player $i$'s strategy in $\strat(\omega)$);
\item[PR2.] $\PR_i(\omega)(\{\omega'\in\Omega :
\PR_i(\omega')=\PR_i(\omega)\})=1$.
\end{itemize}
These assumptions guarantee that player $i$ assigns probability $1$ to
his actual strategy and beliefs.  \wbox
\end{itemize}
}
\end{definition}


We can now define $i$'s beliefs at
$\omega$ if he were to 
switch to strategy $\ssigma'$. Intuitively, if he were to switch to strategy $\ssigma'$ at
$\omega$, the probability that $i$ would assign to state $\omega'$ is
the sum of the probabilities that he assigns to all the states
$\omega''$ such that he believes that he would move from $\omega''$ to
$\omega'$ if he used strategy $\ssigma'$. Thus we define
$$\mathcal P\mathcal R_{i,\ssigma'}(\omega)(\omega'):=\sum_{\{\omega'' :
f(\omega'',i,\ssigma')=\omega'\}}\mathcal P\mathcal R_i(\omega)(\omega''). 
$$

We define the expected utility of player $i$ at state $\omega$ in the
usual way as the sum of the product of his expected utility of the
strategy profile played at each state $\omega'$ and the
probability of $\omega'$:
$\EU_i(\omega)=\sum_{\omega'\in\Omega}\mathcal P\mathcal
R_i(\omega)(\omega')u_i(\mathbf{s}_i(\omega),\mathbf{s}_{-i}(\omega')).$%
\protect{\footnote{Given a profile $t = (t_1, \ldots, t_N)$, as usual, we define
$t_{-i}=(t_1,\ldots,t_{i-1},t_{i+1},\ldots,t_N)$.  We extend this
notation in the obvious way to functions like $\strat$, so that, for
example, 
$\strat_{-i}(\omega) = (\strat_1(\omega), \ldots, \strat_{i-1}(\omega),
\strat_{i+1}(\omega),\ldots, \strat_{n}(\omega))$.}}


Now we define $i$'s expected utility at $\omega$ if he were
to switch to $\ssigma'$. The usual way to do so is to simply replace $i$'s actual strategy at
$\omega$ by $\ssigma'$ at all states, keeping the strategies of the other
players the same; that is,
$$\sum_{\omega'\in\Omega}\PR_i(\omega)(\omega')u_i(\ssigma',\mathbf{s}_{-i}(\omega')).$$
In this definition,  player $i$'s beliefs about the strategies that the
other players are using do not change when he switches from
$\strat_i(\omega)$ to $\ssigma'$.  The key point of counterfactual
structures is that these beliefs may well change. Thus, we define $i$'s
expected utility at $\omega$ if he switches to $\ssigma'$ as
$$
\EU_i(\omega,\ssigma')=\sum_{\omega'\in\Omega}\mathcal
\PR_{i,\ssigma'}(\omega)(\omega')u_i(\ssigma',\mathbf{s}_{-i}(\omega')).  
$$

Finally, we can define rationality in counterfactual structures using
these notions: 
\begin{definition}\label{def:rationality}
{\rm Player $i$ is \emph{rational} at state $\omega$ if, for all
$\ssigma' \in \SSigma_i$,
$$\EU_i(\omega) \ge \EU_i(\omega,\ssigma').$$
\wbox
}
\end{definition}

\subsection{Translucent equilibrium}\label{se:equilibria}
In this section, we define translucent equilibrium and we observe that
the results reported in the main text imply that social dilemmas have
a counterfactual structure according to which each player plays, in
equilibrium, his part of the welfare maximizing strategy with non-zero
probability.

We start with some preliminary notation. Given a probability measure
$\tau$ on a finite set $T$, let 
$\text{supp}(\tau)$ denote the support of $\tau$, that is, $\text{supp}(\tau)=\{t\in T : \tau(t)\neq0\}$. Given a mixed strategy profile $\sigma$, note that $\sigma_{-i}$ can 
can be viewed as a probability on
$S_{-i}$, where
$\sigma_{-i}(s_{-i})=\prod_{j\neq i}\sigma_{j}(s_j)$.
Similarly $\sigma$ can be viewed as a probability measure on $S$.
In the sequel, we view $\sigma_{-i}$ and $\sigma$ as probability
measures without further comment (and so talk about their support).

\begin{definition}\label{def:equilibrium}
{\rm A strategy profile $\sigma$ in a game $\G$ is 
\emph{translucent equilibrium} in a counterfactual structure $M =
(\Omega, \mathbf{s}, f, \PR_1,\ldots,\PR_N)$ 
appropriate for $\G$ if there exists a subset $\Omega' \subseteq \Omega$
such that, for each state $\omega$ in $\Omega'$, the following
properties hold:
\begin{enumerate}
\item[TE1.] $\mathbf{s}(\omega)\in\text{supp}(\sigma)$;
\item[TE2.] $\text{supp}(\PR_i(\omega))\subseteq\Omega'$;
\item[TE3.] $\mathbf{s}_{-i}(\PR_i(\omega))=\sigma_{-i}$ (i.e.,
for each strategy profile $s_{-i} \in S_{-i}$, we have 
$\sigma_{-i}(s_{-i}) = \PR_i(\omega)(\{\omega': \strat_{-i}(\omega') =
s_{-i}\})$). 
\item[TE4.] each player is rational at $\omega$.
\end{enumerate}
The mixed strategy profile $\sigma$ is a translucent equilibrium of $\G$
if there exists a counterfactual structure $M$ appropriate for $\G$ such
that $\sigma$ is a translucent equilibrium in $M$.}  \wbox
\end{definition}
Intuitively, $\sigma$ is a translucent equilibrium in $M$ if, for each
strategy $s_i$ in the support of $\sigma_i$, the expected utility of
playing $s_i$ given that other players are playing according to
$\sigma_{-i}$ is at least as good as switching to some other strategy
$s_i'$, given what $i$ would believe about what strategies the other
players are playing if he were to switch to $s_i'$.  

This notion of translucent equilibrium is closely related to a
condition called \emph{IR} (for \emph{individually rational}) by
Halpern and Pass \citeyear{HaPa13}.  The main difference is that
Halpern and Pass considered only pure strategy profiles; we allow
mixed-strategy profiles here.  We discuss the relationship 
between the notions at greater length in Section~\ref{sec:charte}.

\commentout{
Let now $\G$ be a social dilemma and recall that 
$s^N$ and $s^W$ denote respectively the unique Nash equilibrium and
the unique welfare-maximizing profile. 
Let $\G$ be a social dilemma and recall that 
$s^N$ and $s^W$ denote the unique Nash equilibrium and
the unique welfare-maximizing profile, respectively. 
We define a parameterized family of counterfactual structures
$M^\G(\sigma,\alpha) = 
(\Omega,\strat,\PR_1,\ldots,\PR_N)$, where $\alpha = 
(\alpha_1,\ldots, \alpha_N)$, $\alpha_i \in [0,1]$.
Intuitively, $\alpha_j$ describes how likely player $j$ is to
learn about a deviation (i.e., how ``translucent'' the other players
are to $j$) and $\sigma=(\sigma_{-i})_{i}$ describes what player $i$ believes the other players are going to do. We assume that if $j$ detects a deviation on the part of
any player, then she will play $s_j^N$. We also assume that the beliefs are supported on profiles of strategies $(s_1,\ldots,s_N)$, where $s_i\in\{s_i^N,s_i^W\}$. Moreover, we assume that $\sigma_{-i}$ is a power measure, in the sense that, for every $i,j$, with $i\ne j$, player $i$ believes that Player $j$ will be playing the strategy $\sigma_{i,j}:=p_{i,j}s_{j}^W+(1-p_{i,j})s_j^N$, where 
\begin{enumerate}
\item $p_{i,j}=:p_{-i}$ does not depend on $j$.
\item $\sigma_{-i}=\otimes_{j\ne i}\sigma_{i,j}$.
\end{enumerate} 
We define
$M^\G(\sigma,\alpha)$ as follows.
\begin{itemize}
\item $\Omega = S\times \{0,1\}^{n}$.  (The second component of the
state, which is an element of $\{0,1\}^{n}$, is used to determine the
closest-state function.  Roughly speaking, if $v_j = 1$, then player
$j$ learns  about a deviation if there is one; if $v_j = 0$, he does not.)
\item $\strat((s,v)) = s$.
\item $f((s,v),i,s_i^*) = 
\left\{
\begin{array}{ll}
(s,v) &\mbox{if $s_i = s_i^*$,}\\
(s',v) &\mbox{if $s_i \ne s_i^*$, where $s'_i = s_i^*$ and for $j
\ne i$,}\\
&\mbox{$s'_j = s_j$ if $v_j = 0$ and $s'_j = s^N_j$ if $v_j = 1$.}
\end{array}
\right.$

\noindent Thus, if player $i$ changes strategy from $s_i$ to $s_i'$, $s_i'\neq
 s_i$, then each other player $j$ either deviates to his component of
the Nash equilibrium or continues with his current strategy, 
depending on whether $v_j$ is 0 or 1. Roughly speaking, he switches to his
component of the Nash equilibrium if he learns about a deviation
(i.e., if $v_j = 1$).
\item $\PR_i(s,v)(s',v') =\left\{
\begin{array}{lll}
0 &\mbox{if $s_i \ne s'_i$ or $v_i \ne v_i'$,}\\
\sigma_{-i}(s'_i)\Pi_{\{j \ne i: v_j = 1\}}\alpha_j\Pi\, _{\{j \ne
i: \, v_j = 0\}}(1-\alpha_j) &\mbox{otherwise.}
\end{array}
\right.$
\noindent Thus, the probability of the $s'$ component of the state is
determined 
by $\sigma$ (with the standard assumption that player $i$ knows his
own strategy).  The probability of the $v'$ component is determined by
assuming that each player $j$ independently learns about a deviation
by $i$ with probability $\alpha_j$.\footnote{For simplicity, we are
  assuming that the probability that $j$ learns about a deviation by
  $i$ is the same for each player $i$.  More generally, we could have
  a probability $\alpha_{ji}$ for the probability that $j$ learns
  about a deviation by $i$.  Nothing significant would change in our analysis.}
\end{itemize}

Consider now the Prisoner's dilemma and, using a notation consistent with the main text, assume that Player 1 believes that Player 2 plays the strategy $\beta_2C+(1-\beta_2)D$ and that Player 2 believes that Player 1 plays the strategy $\beta_1C+(1-\beta_1)D$. If $\alpha_ib\beta_{-i}\geq c$, for every $i$, then Proposition \ref{prop:PD} implies that cooperation is rational for both players. Thus, if we define $\Omega'=\{((C,C),(1,1)),((C,D),(1,1)),((D,C),(1,1)),((D,D),(1,1))\}$, then all conditions in Definition \ref{def:equilibrium} are satisfied with $\sigma=(\beta_1,\beta_2)$ and so it is a translucent equilibrium. Similarly, it is not hard to see how the other propositions discussed in the main text imply that there are mixed translucent equilibria where each player plays his part of the welfare-maximizing strategy with non-zero probability.
}

\subsection{Characterization of translucent equilibria}\label{sec:charte}

While it is easy to see that every Nash equilibrium is a translucent
equilibrium (see Proposition \ref{prop:nashistranslucent}), the
converse is far from true.  As we show, for example, cooperation can
be an equilibrium in social dilemmas (see below and Section~\ref{sec:teinsd}). 
In this section, we provide a characterization of translucent
equilibria that will prove useful when discussing social dilemmas.


\begin{proposition}\label{prop:nashistranslucent}
 Every Nash equilibrium of $\G$ is a translucent
equilibrium.
\end{proposition}
\begin{proof} 
Given a Nash equilibrium 
$\sigma=(\sigma_1,\ldots,\sigma_n)$, consider the following
counterfactual structure $M_\sigma = (\Omega, \mathbf{s}, f,
\PR_1,\ldots,\PR_N)$:
\begin{itemize}

\item $\Omega$ is the set of strategy profiles in the support of
$\sigma$;
\item $\strat(s) = s$;
\item $\PR_i(s_i,s_{-i})(s_i',s_{-i}') = 
\left\{
\begin{array}{ll}
0 &\mbox{if $s_i' \ne s_i$}\\
\sigma_{-i}(s_{-i}') &\mbox{if $s_i' = s_i$;}
\end{array}
\right.$
\item $f((s_i,s_{-i}),i,s')=(s',s_{-i})$. 
\end{itemize}

It is easy to check that $\sigma$ is a translucent equilibrium in
$M_\sigma$; we simply take $\Omega' = \Omega$.  The fact that $f$ is an
``opaque'' closest-state function, which is not affected by the strategy
used by players, means that rationality in $M$ reduces to the standard
definition of rationality.  We leave details to the reader.  
\end{proof}

Although the fact that we can consider arbitrary counterfactual
structures (appropriate for $\G$) means that many strategy profiles are
translucent equilibria, the notion of translucent equilibrium has some
bite.  For example, the strategy profile $(C,D)$, where player 1
cooperates and player 2 defects, is not a translucent
equilibrium in Prisoner's Dilemma: if player 1 believes that player 2
is playing defecting with probability 1, there are no beliefs that 1
could have that would justify cooperation.  However, as we shall see,
both $(C,C)$ and $(D,D)$ are translucent equilibria.  This follows from
the characterization of translucent equilibrium that we now give.

\begin{definition}\label{dfn:coherent}
{\rm A mixed-strategy profile $\sigma$ in $\G$ is \emph{coherent} if
for all players $i \in P$, all $s_i\in\text{supp}(\sigma_i)$, and all
$s_i'\in S_i$, there is $s_{-i}'\in S_{-i}$ such that 
$$u_i(s_i,\sigma_{-i})\geq u_i(s')$$
(where, of course, $u_i(s_i,\sigma_{-i}) = \sum_{s_{-i}'' \in S_{-i}'}
\sigma_{-i}(s_{-i}'') u_i(s_i, s_{-i}'')$).
} \wbox
\end{definition}
That is, $\sigma$ is coherent if, for all pure strategies for player $i$
in the support of $\sigma_i$, if $i$'s belief about the strategies being
played by the other  players is given by $\sigma_{-i}$, 
 there is no obviously better strategy that
$i$ can switch to in the weak sense that, if $i$ contemplates switching
to $s_i'$, there are beliefs that $i$ could have about the other players
(namely, that they would definitely play $s_{-i}'$ in this case) that
would make switching 
to $s_i'$ better than sticking with $s_i$.

It is easy to see that $(C,C)$ and $(D,D)$ in Prisoner's Dilemma are
both coherent;  on the other hand, $(C,D)$ is not.  

Halpern and Pass \citeyear{HaPa13} define a pure strategy profile to
be \emph{individually rational} if it is coherent.
Definition~\ref{dfn:coherent} extends individual rationality to mixed
strategies.  Halpern and Pass prove that a pure strategy profile
is
individually rational 
if there is a model where it is commonly known that $\sigma$ is played
and there is common belief of rationality.  The definition of
translucent equilibrium can be seen as the generalization of this
characterization of IR to mixed strategies.  As the following theorem
shows, we get an analogous representation.

\begin{theorem}\label{thm:translucenteq}
The mixed strategy profile $\sigma$ of game $\G$ is coherent iff $\sigma$ is a
translucent equilibrium of $\G$.
\end{theorem}

\begin{proof}
Let $\sigma$ be a coherent strategy profile in $\G$.  We construct a counterfactual
structure $M = (\Omega, \mathbf{s}, f, \PR_1,\ldots,\PR_N)$ as follows:
\begin{itemize}
\item $\Omega = S$;
\item $\strat(s) = s$;
\item $\PR_i(\omega)(\omega') = \left\{
\begin{array}{ll}
1 &\mbox{if $\omega \notin \supp(\sigma_i)$, $\omega = \omega'$}\\
0 &\mbox{if $\omega \notin \supp(\sigma_i)$, $\omega \ne \omega'$}\\
\sigma_{-i}(\strat_{-i}(\omega')) &\mbox{if $\omega \in \supp(\sigma_i)$,
$\strat_i(\omega') = \strat_i(\omega)$}\\
0 &\mbox{if $\omega \in \supp(\sigma_i)$,
$\strat_i(\omega') \ne \strat_i(\omega)$};
\end{array}
\right.$
\item $f(\omega,i,s_i') = \left\{
\begin{array}{ll}
(s_i',\mathbf{s}_{-i}(\omega)) &\mbox{if $\omega \notin
\supp(\sigma_i)$}\\
\omega &\mbox{if $\omega \in \supp(\sigma_i)$, $s_i' = \strat_i(\omega)$}\\
(s_i',s_{-i}') & \mbox{if $\omega \in \supp(\sigma_i)$, $s_i' \ne
\strat_i(\omega)$, where $s_{i}'$ is a}\\ &\mbox{strategy such that
$u_i(\mathbf{s}_i(\omega),\sigma_{-i})\geq u_i(s')$;} \\
&\mbox{such a strategy is
guaranteed to exist since}\\ &\mbox{$\sigma$ is coherent.} 
\end{array}
\right.$
\end{itemize}

We first show that $M$ is a finite counterfactual structure appropriate
for $\G$; in particular, $\PR_i$ satisfies PR1 and PR2 and $f$
satisfies CS1 and CS2.  For PR1 and PR2, there are two cases.  If
$\omega \notin \supp(\sigma)$, then $\PR_i(\omega)(\omega) = 1$, so PR1
and PR2 clearly hold.  If $\omega \notin \supp(\omega)$, then 
$\PR_i(\omega)(\omega) > 0$ iff $\strat_i(\omega) = \strat_i(\omega')$.
Moreover, if $\strat_i(\omega) = \strat_i(\omega')$, then it is
immediate from the definition that $\PR_i(\omega) = \PR_i(\omega')$, so
PR2. holds.  That CS1 and CS2 hold is immediate from the definition of $f$.

To show that $\sigma$ is a translucent equilibrium in $M$, let
$\Omega' = \supp(\sigma)$.  For each state $\omega \in \Omega'$, TE1
clearly holds.  Note that if $\omega \in
\supp(\sigma)$, then $\PR_i(\omega) = (\strat_i(\omega),
\sigma_{-i}(\omega))$ (identifying the strategy profile with a
probability measure), so TE2 and TE3 clearly hold. 
It remains to show that TE4 holds, that is, 
that every player is rational at every state $\omega\in\Omega'$.

Thus, we must show that $\EU_i(\omega) \ge \EU(\omega,s_i^*)$ for all
$s_i^* \in S_i$.  Note that
$$\begin{array}{lll}
\EU_i(\omega)
&= &\sum_{\omega'\in\Omega}\mathcal
\PR_i(\omega)(\omega')u_i(\mathbf{s}_i(\omega),\mathbf{s}_{-i}(\omega')) \\
&= &\sum_{\{\omega'\in\Omega: \strat_i(\omega') = \strat_i(\omega)\}}
\sigma_{-i}(\strat_{-i}(\omega'))
u_i(\mathbf{s}_i(\omega),\mathbf{s}_{-i}(\omega'))\\ 
&= &\sum_{s_{-i}''\in S_{-i}}
u_i(\mathbf{s}_i(\omega),s_{-i}'')\\ 
&=&u_i(\strat_i(\omega),\sigma_{-i}).
\end{array}
$$
By definition,
$$
\EU_i(\omega,\ssigma_i^*)=\sum_{\omega'\in\Omega}\mathcal
\PR_{i,\ssigma_i^*}(\omega)(\omega')u_i(\ssigma_i^*,\mathbf{s}_{-i}(\omega'))$$
and 
$$\PR_{i,\ssigma'}(\omega)(\omega')=\sum_{\{\omega'' :
f(\omega'',i,\ssigma')=\omega'\}}\mathcal P\mathcal R_i(\omega)(\omega''). 
$$
Now if 
$s_i^*=\mathbf{s}_i(\omega)$, then $f(\omega,i,s_i^*)$.  In this case,
it is easy to check that $\PR_{i,\ssigma_i^*}(\omega) = \PR_i(\omega)$,
so 
$\EU_i(\omega,s_i^*)=\EU_i(\omega) = \EU_i(s_i,\sigma_{-i})$, and TE4
clearly holds.  
On the other hand, if 
$s_i^*\neq\mathbf{s}_i(\omega)$, then 
$$\begin{array}{lll}
\EU_i(\omega,s_i^*)
&=&\sum_{\omega'\in\Omega}\sum_{\{\omega'':f(\omega'',i,s_i^*)=\omega'\}}\PR_i(\omega)(\omega'')u_i(s_i^*,\mathbf{s}_{-i}(\omega'))\\
&=&\sum_{\{\omega'\in\Omega: \strat_i(\omega') =
s_i^*\}}\sum_{\{\omega'':f(\omega'',i,s_i^*)=\omega',\, 
\strat_i(\omega'') = \strat_i(\omega)\}}\sigma_{-i}(\omega'')
u_i(s_i^*,\mathbf{s}_{-i}(\omega'))\\ 
&=&\sum_{\{\omega'\in\Omega: \strat_i(\omega') =
s_i^*\}}\sum_{\{\omega'':f(\omega'',i,s_i^*)=\omega',\, 
\strat_i(\omega'') = \strat_i(\omega)\}}\sigma_{-i}(\omega'')
u_i(f(\omega'',i,s_i^*)).\\ 
\end{array}
$$
By definition, $u_i(f(\omega'',i,s_i^*))\leq
u_i(\strat_i(\omega''),\sigma_{-i})=u_i(\strat_i(\omega),\sigma_{-i})$. Thus,
$$\begin{array}{lll}
\EU_i(\omega,s_i^*)
&\le&\sum_{\{\omega'\in\Omega: \strat_i(\omega') =
s_i^*\}}\sum_{\{\omega'':f(\omega'',i,s_i^*)=\omega',\, 
\strat_i(\omega'')= \strat_i(\omega)\}}\sigma_{-i}(\omega'')
u_i(\strat_i(\omega),\sigma_{-i})\\ 
&=&u_i(\strat_i(\omega),\sigma_{-i}) \sum_{\{\omega'\in\Omega: \strat_i(\omega') =
s_i^*\}}\sum_{\{\omega'':f(\omega'',i,s_i^*)=\omega',\, 
\strat_i(\omega'') = \strat_i(\omega)\}}\sigma_{-i}(\omega'')\\
&=&u_i(\strat_i(\omega),\sigma_{-i}).
\end{array}
$$
This completes the proof that TE4 holds, and the proof of the ``only if''
direction of the argument

The ``if'' is actually much simpler. Suppose, by way of contradiction,
that $\sigma$ is not coherent.  Then
there is a player $i$ and a strategy $s_i\in\text{supp}(\sigma_i)$ such
that for all $s_{-i}' \in S_i$, we have $u_i(s_i,\sigma_{-i})<u_i(s')$. It
follows that, for all counterfactual structures $M$, no matter what the
beliefs and the closest-state functions are in $M$,  
it is always strictly profitable for player $i$ to switch strategy from
$s_i$ to $s_i'$. Consequently, $i$ is not rational at a state $\omega$
such that $s_i(\omega)=s_i$, contradicting TE4.
\end{proof}

\commentout{
We conclude by comparing translucent equilibrium to iterated minimax
domination.  We begin by reviewing the relevant definitions.

\begin{definition} {\rm Strategy $s_i$ for player $i$ in game $\G
= (P,\SSigma_1,\ldots,\SSigma_N,u_1,\ldots,u_N)$ is 
\emph{minimax dominated with respect to $S'_{-i}\subseteq S_{-i}$} iff
there 
exists a strategy $s'_i \in S_i$ such that  
$$\min_{t_{-i} \in
S'_{-i}} u_i(s'_i, t_{-i}) > \max_{t_{-i} \in 
S'_{-i}} u_i(s_i, t_{-i}).$$}
\wbox
\end{definition}

\noindent Thus, player $i$'s strategy $s_i$ is minimax dominated
with respect to $S'_{-i}$ iff there exists
a strategy $s_i'$ for player $i$ such that the worst-case payoff for
player $i$ if he uses $s'$ is strictly better than his best-case
payoff if he uses $s$, given that the other players are restricted to using a
strategy in $S'_{-i}$.

\begin{definition} Given a game $\G =
(P,\SSigma_1,\ldots,\SSigma_N,u_1,\ldots,u_N)$,  
define $\NSD_j^k(\G)$ inductively: let $\NSD_j^0(\G) =
S_j$ and let $\NSD_j^{k+1}(\G)$
consist of the strategies in $\NSD_j^{k}(\G)$ not minimax 
dominated with respect to $\NSD_{-j}^{k}(\G)$.
Strategy $s \in S_i$ \emph{survives iterated minimax domination}
if $s \in \cap_k\NSD_i^k(\G)$.  \wbox
\end{definition}

As these definitions show, only pure strategies of individual players
survive (or do not survive) iterated minimax domination.  It is not hard
to show that both $C$ and $D$ survive iterated minimax domination.  But,
as we have observed, although $(C,C)$ and $(D,D)$ are translucent
equilibria in Prisoner's Dilemma, $(C,D)$ is not.  So not every profile
made up of strategies that survive iterated minimax domination is a
translucent equilibrium.  On the other hand, as the following example
shows, the strategies in a profile that is a translucent equilibrium may
not survive iterated minimax domination.  

\begin{example} {\rm Consider the two-player game where $S_1 =
\{s_1^1,s_2^1,s_3^1\}$, $S_2 = \{s_1^2,s_2^2,s_3^2\}$, and the utilities
are defined by the following table:
\begin{table}[htb]
\begin{center}
\begin{tabular}{|| c | c | c | c ||}
\hline
{} & $s_1^2$ & $s_2^2$ & $s_3^2$\\
\hline
$s_1^1$ & (1,1) & (1,2) & (1,3) \\
\hline
$s_2^1$ & (2,1) & (2,2) & (2,3) \\
\hline
$s_3^1$ & (3,1) & (3,2) & (3,3)) \\
\hline
\end{tabular}
\end{center}
\end{table}
Thus, if player $i$ plays $s_i^k$, then his utility is $k$, independent
of what the other player does.

It easily follows from Theorem~\ref{thm:translucenteq} that
$(s_1^2,s_2^2)$ is a translucent equilibrium, since it is coherent (we
can take $s'$ in Definition~\ref{dfn:coherent} to be $(s_1^1,s_1^2)$.
On the other hand, $s_1^j$ and $s_2^j$ are minimax dominated by $s_3^j$.
\wbox }
\end{example}
}

\subsection{Translucent equilibrium in social dilemmas}\label{sec:teinsd}

As we now show, our characterizations of
Propositions~\ref{prop:PD}--\ref{pro:Bertrand} can be used to provide
conditions on when translucent equilibrium exists in these social
dilemmas.

We start our analysis with Prisoner's Dilemma. We capture the
assumption that $\beta$ is the probability of cooperation, and that
players either cooperate or defect, by assuming
that players follow a mixed strategy where they cooperate with
probability $\beta$ and defect with probability $1-\beta$.

\begin{proposition}\label{pro:PDequilibrium1} 
$(\beta_1 C + (1-\beta_1)D, \beta_2 C + (1-\beta_2)D)$ is a
  translucent equilibrium of Prisoner's dilemma iff
$\beta_i b \ge c$, for $i = 1,2$, or $\beta_1 = \beta_2 = 0$.  
\end{proposition}

\begin{proof} Suppose that $(\beta_1 C + (1-\beta_1)D, \beta_2 C +
  (1-\beta_2)D)$.  If $\beta_1 > 0$, then
by Theorem~\ref{thm:translucenteq}, it easily follows
we must have $u_1(C,\beta_2 C + (1-\beta_2)D) \ge u_1(D,D)$.
Thus, we must have $\beta_2(b-c) + (1-\beta_2)(- c) \ge 0$;
equivalently, $\beta_2 b \ge c$.  Note that since $c > 0$, this means
that we must have $\beta_2 > 0$.  Similarly, if $\beta_2  > 0$, then
$\beta_1 b \ge c$.  By Theorem~\ref{prop:nashistranslucent}, $(D,D)$
is a translucent   equilibrium, since it is a Nash equilibrium.  
Thus, either $\beta_i b \ge c$ for $i = 1,2$ or $\beta_1 = \beta_2 =
0$.  

Conversely, if $\beta_i b \ge c$ for $i = 1, 2$, then it again easily
follows from Theorem~\ref{thm:translucenteq} that 
$(\beta_1 C + (1-\beta_1)D, \beta_2 C + (1-\beta_2)D)$ is a
  translucent equilibrium.  As we have observed, $(D,D)$ (the case
  that $\beta_1 = \beta_2 = 0$) is also a translucent equilibrium.
\end{proof}

Proposition~\ref{pro:PDequilibrium1} 
is not all that interesting, since it does not take into account
a player's beliefs regarding translucency.  The following definition
is a step towards doing this.
Suppose that $M$ is  counterfactual structure appropriate for a social
dilemma $\Gamma$.  
\emph{Player $i$ has type $\alpha_i$ in $M$} if, at each
state $\omega$ in $M$, player $i$ believes that 
if he intends to cooperate in $\omega$ and deviates from that,
then each other  agent will independently realize this with
probability $\alpha_i$ and will defect.
Formally, this means that,
at each state $\omega$ in $M$, we have 

\commentout{
\begin{itemize}
\item if $\strat_i(\omega) =
\ssigma_i^W$ (i.e., $i$ is cooperating in $\omega$ by playing his component of the
social-welfare maximing strategy profile),  then, for each player
$j\ne i$, we have $\PR_i(\omega)(\{\omega': f(\omega',i,s_i^N) =
\omega'', s_j(\omega') =s_j^C, s_j(\omega'') = s_j^N\}) = 
\alpha_i \PR_i(\omega)\{\omega': s_j(\omega') = s_j^C\})$.
\end{itemize}
}
\begin{itemize}
\item if $\strat_i(\omega) =
\ssigma_i^W$ (i.e., $i$ is cooperating in $\omega$ by playing his component of the
social-welfare maximing strategy profile),  then, for each $J\subseteq P\setminus\{i\}$, we have
$\PR_i(\omega)(\{\omega': f(\omega',i,s_i^N) =
\omega'', s_j(\omega') =s_j^C, s_j(\omega'') = s_j^N, \forall j\in J\}) = 
\alpha_i ^{|J|}\PR_i(\omega)\{\omega': s_j(\omega') = s_j^C, \forall j\in J\})$.
\end{itemize}


\begin{proposition}\label{pro:PDequilibrium}
$(\beta_1 C + (1-\beta_1)D, \beta_2 C + (1-\beta_2)D)$ is a
  translucent equilibrium of the Prisoner's dilemma in a structure where player $i$ has type
  $\alpha_i$ 
if and only if $\beta_1 = \beta_2 = 0$ or
$\alpha_i \beta_{3-i} b \ge c$ for $i \ge 1,2$.
\end{proposition}

\begin{proof}
Suppose that 
$\alpha_i \beta_{3-i} b \ge c$ for $i = 1 ,2$ or $\beta_1 = \beta_2 = 0$.
  We show that $(\beta_1 C + (1-\beta_1)D,
\beta_2 C + (1-\beta_2)D)$ is a
translucent equilibrium in a structure where 
player $i$ has type  $\alpha_i$. 
Consider the counterfactual structure
$M(\alpha_1,\alpha_2)$ defined as follows
\begin{itemize}
\item $\Omega = \{C,D\}\times \{0,1\}^{2}$.  (The second component of the
state, which is an element of $\{0,1\}^{2}$, is used to determine the
closest-state function.  Roughly speaking, if $v_j = 1$, then player
$j$ learns  about a deviation if there is one; if $v_j = 0$, he does not.)
\item $\strat((s,v)) = s$.
\item $f((s,v),i,s_i^*) = 
\left\{
\begin{array}{ll}
(s,v) &\mbox{if $s_i = s_i^*$,}\\
(s',v) &\mbox{if $s_i \ne s_i^*$, where $s'_i = s_i^*$ and for $j
\ne i$,}\\
&\mbox{$s'_j = s_j$ if $v_j = 0$ and $s'_j = s^N_j$ if $v_j = 1$.}
\end{array}
\right.$

\noindent Thus, if player $i$ changes strategy from $s_i$ to $s_i'$, $s_i'\neq
 s_i$, then each other player $j$ either deviates to his component of
the Nash equilibrium or continues with his current strategy, 
depending on whether $v_j$ is 0 or 1. Roughly speaking, he switches to his
component of the Nash equilibrium if he learns about a deviation
(i.e., if $v_j = 1$).
\item $\PR_i(s,v)(s',v') =\left\{
\begin{array}{lll}
0 &\mbox{if $s_i\ne s_i'$, or $v_i \ne v_i'$,}\\
\sigma_{3-i}(s_{3-i})\pi_i(v_{3-i})  &\mbox{if $s = s'$},
\end{array}
\right.$
where $\sigma_{3-i}$ is the distribution on strategies that puts
probability $\beta_{3-i}$ on $C$ and probability $1-\beta_{3-i}$ on
$D$, while $\pi_i$ is the distribution that puts probability $\alpha_i$
on 1 and probability $1-\alpha_i$ on 0.
\noindent Thus, if $s = s'$, then the 
probability of the $v'$ component is determined by
assuming that the other player ($3-i$) independently learns about a deviation
by $i$ with probability $\alpha_i$.
\end{itemize}

Clearly, $M(\alpha_1,\alpha_2)$ is a structure where player $i$ has type
$\alpha_i$, for $i=1,2$. 
We claim that $(\beta_1  C + (1-\beta_1)D, \beta_2 C + (1-\beta_2)D)$
is a translucent equilibrium in the 
counterfactual structure $M(\alpha_1,\alpha_2)$.  

There are two cases.  If $\beta_1 = \beta_2 = 0$, then 
let $\Omega'$ consist of all states of the form $((D,D),v)$.  It is
easy to check that TE1--4 hold.
If $\alpha_i \beta_{3-i} b \ge c$ for $i \ge 1,2$,
let $\Omega' = \Omega$. 
It is immediate that TE1, TE2, and TE3 hold.  Since $\alpha_i \beta_{3-i} b
\ge c$, it follows from Proposition~\ref{prop:PD} that player $i$ is
rational at each state in $\Omega$; thus, TE4 holds.  

For the converse, suppose that $M$ is a structure where player $i$ has type
$\alpha_i$, for $i=1,2$, and  
$(\beta_1  C + (1-\beta_1)D, \beta_2 C + (1-\beta_2)D)$ is a
translucent equilibrium in $M$.    If it is not the case that either
$\beta_1 = \beta_2 = 0$ or $\alpha_i \beta_{3-i} b \ge c$ for $i=1,2$,
without loss of generality we can assume that $\beta_1 > 0$ and 
that $\alpha_1 \beta_2 b < c$.  Let $\omega$ be a state in the set
$\Omega'$ where player 1 cooperates.  Since player 1 must  be rational
at $\Omega'$, we must have $u_1(C,\beta_2 C + (1-\beta_2) D) \ge 
((1-\beta_2) + \alpha_1\beta_2) u_1(D,D) + (1-\alpha_1)\beta_2 u_1(D,C)$.
Simple calculations show that this inequality holds iff
$\beta_2 (b-c) + (1-\beta_2)(-c) \ge (1-\alpha_1)\beta_2 b$ or, equivalently,
$\alpha_1 \beta_2 b \ge c.$  This gives the desired contradiction.
\end{proof}


The following propositions can be proved in a similar fashion. We leave details to the reader.

\begin{proposition}\label{pro:TDequilibrium1} 
$(\beta_1 H + (1-\beta_1)L, \beta_2 H + (1-\beta_2)L)$ is a
translucent equilibrium of the Traveler's dilemma if and only if 
$b\le\frac{(H-L)\beta_i}{1-\beta_i}$, for $i = 1,2$, or $\beta_1 = \beta_2 = 0$.  \wbox
\end{proposition}

\begin{proposition}\label{pro:TDequilibrium}
$(\beta_1 H + (1-\beta_1)L, \beta_2 H + (1-\beta_2)L)$ is a
  translucent equilibrium of the Traveler's dilemma in a structure where player $i$ has type
  $\alpha_i$ 
if and only if $\beta_1 = \beta_2 = 0$ or
$$b\le
\left\{
\begin{array}{ll}
\frac{(H-L)\beta_{3-i}}{1-\alpha_i\beta_{3-i}} &\mbox{if $\alpha_i\ge\frac12$}\\
\min\left(\frac{(H-L)\beta_{3-i}}{1-\alpha_i\beta_{3-i}},\frac{H-L-1}{1-2\alpha_i}\right)
&\mbox{if $\alpha_i<\frac12$.} 
\end{array}
\right.$$ \wbox
\end{proposition}

In the following propositions, let C and D denote, respectively, the
full contribution and the null contribution in the Public Goods
game. Given an $N$-tuple $(r_1,\ldots,r_N)$ of real numbers,
$\bar r_{-i}$ denotes the average of the numbers $r_j$, with $j\ne i$. 

\begin{proposition}\label{pro:PGGequilibrium1} 
$(\beta_1 C + (1-\beta_1)D,\ldots, \beta_N C + (1-\beta_N)D)$ is a
translucent equilibrium of the Public Goods game if and only if 
$\rho\bar\beta_{-i}(N-1)\ge1-\rho$ for all $i$, or $\beta_i=0$ for
all $i$.  \wbox  
\end{proposition}

\begin{proposition}\label{pro:PGGequilibrium}
$(\beta_1 C + (1-\beta_1)D, \ldots,\beta_N C + (1-\beta_N)D)$ is a
  translucent equilibrium of the Public Goods game in a structure where player $i$ has type
  $\alpha_i$ 
if and only if $\beta_i=0$ for all $i$ or
$\alpha_i\rho\bar\beta_{-i}\ge1-\rho$ for all $i$.\wbox 
\end{proposition}

\begin{proposition}\label{pro:BCequilibrium} 
$(\beta_1 H + (1-\beta_1)L,\ldots, \beta_N H + (1-\beta_N)L)$ is a
translucent equilibrium of the Bertrand competition if and only if
$\beta_i= 0$ for all $i$, or 
$\prod_{j\ne i}\beta_j\ge\frac{L}{H}$ for all $i$.  \wbox
\end{proposition}

\commentout{
\begin{proposition}\label{pro:BCequilibrium1}
$(\beta_1 H + (1-\beta_1)L, \ldots,\beta_N H + (1-\beta_N)L)$ is a
  translucent equilibrium of the Bertrand competition in a structure where player $i$ has type
  $\alpha_i$ 
if and only if $\beta_i = 0$ for all $i$, or $\prod_{j\ne
  i}\beta_j\ge\frac{LN}{H((1-\gamma_i)(N-1)+1)}$ for all $i$, where
$\gamma_i = (1-\alpha)\bar\beta_{-i}$.\wbox 
\end{proposition}
}
\begin{proposition}\label{pro:BCequilibrium1}
$(\beta_1 H + (1-\beta_1)L, \ldots,\beta_N H + (1-\beta_N)L)$ is a
  translucent equilibrium of the Bertrand competition in a structure where player $i$ has type
  $\alpha_i$ if and only if $\beta_i = 0$ for all $i$, or $\prod_{j\ne
  i}\beta_j\ge f(\gamma_{i,j},N)LN/H$ for all $i$, where
  $f(\gamma_{i,j},N)=\sum_{J\subseteq P - \{i\}}(\prod_{j\in
    P- (J \cup \{i\})}\gamma_{i,j}\prod_{j\in
    J}(1-\gamma_{i,j}))/(|J|+1)$ and 
$\gamma_{i,j} = (1-\alpha_i)\beta_{j}$.\wbox 
\end{proposition}

\paragraph{Acknowledgments:}  We thank Krzysztof Apt for useful
comments on an earlier draft of this paper.  
Joseph Halpern was supported in part by NSF grants 
IIS-0911036 and  CCF-1214844, AFOSR grant FA9550-08-1-0438, ARO grant
W911NF-14-1-0017, and by the DoD 
Multidisciplinary University Research Initiative (MURI) program administered by AFOSR under grant FA9550-12-1-0040.  Valerio Capraro was funded by the Dutch Research
Organization (NWO) grant 612.001.352.




 \bibliographystyle{chicago}
\bibliography{joe}



%

\end{document}